\documentclass[12 pt]{article}
\usepackage{graphicx} 
\usepackage{tikz}
\usepackage{wrapfig}
\usetikzlibrary{arrows.meta, positioning, calc}
\usepackage{tikz-3dplot}
\usepackage{pgfplots}
\usepackage{subcaption}
\usepackage{color}
\usepackage{amsmath,amssymb,amsthm, mathtools}
\usepackage{natbib}
\usepackage{setspace}
\usepackage{aliascnt}
\usepackage{nameref}
\usepackage[hidelinks]{hyperref} 
\usepackage{xspace}
\usepackage{cleveref}
\usepackage[top=1.25in, bottom=1.25in, left=1.25in, right=1.25in]{geometry}
\usepackage{comment}

\newtheorem{theorem}{\sc Theorem}
\newtheorem*{theorem*}{\sc Theorem}
\newtheorem{lemma}{\sc Lemma}
\newtheorem{proposition}{\sc Proposition}

\newtheorem{definition}{\sc Definition}

\newtheorem*{axiom:iia}{\sc Independence of Irrelevant Alternatives (IIA)}
\newtheorem*{axiom:ie}{\sc Ignorance Equivalence (IE)}

\newcommand{\iia}{\hyperref[axiom:iia]{IIA}\xspace}
\newcommand{\ie}{\hyperref[axiom:ie]{IE}\xspace}

\DeclareMathOperator{\aff}{aff}
\DeclareMathOperator{\dom}{dom}
\DeclareMathOperator{\cav}{cav}
\DeclareMathOperator{\ri}{ri}
\DeclareMathOperator{\epi}{epi}

\colorlet{domainblue}{blue!50}

\title{\bf{Posterior-Separable Costs and Menu Preferences}%
\footnote{This paper partially subsumes results from 
\cite{de2014axiomatic}, in which our two main axioms originally
appear along with a partial characterization of the information
costs that are consistent with them. We would like to thank 
Tommaso Denti for his helpful comments and suggestions, as well as 
seminar participants at SAET 2025. We also thank Victor Luccas Ramalho Moura for his excellent research assistance.}}
\author{Henrique de Oliveira%
\footnote{\texttt{henrique.oliveira@fgv.br, https://henriquedeoliveira.com}}\\ \centering \it \small
FGV - S\~{a}o Paulo School of Economics
\and Jeffrey Mensch%
\footnote{\texttt{jeffrey.mensch@mail.huji.ac.il, https://sites.google.com/site/jeffreyimensch/.}} \\ \centering \it \small Hebrew University of Jerusalem}
\date{\today}

\begin{document}

\maketitle
\begin{abstract}

\noindent 

We consider an agent with a rationally inattentive preference over menus of acts, as in \cite{de2017rationally}. We show that two axioms, \emph{Independence of Irrelevant Alternatives} and \emph{Ignorance Equivalence}, are necessary and sufficient for this agent to have a posterior-separable cost satisfying a mild smoothness condition, called joint-directional differentiability. Viewing the decision-maker's problem as a Bayesian persuasion problem, we also show that these axioms are necessary and sufficient for solvability by a unique hyperplane. When the cost function remains invariant for different priors, we show that these axioms
imply uniformly posterior separable costs that are differentiable. 

\end{abstract}
\section{Introduction}

In many economic settings, a decision maker (DM) must acquire information
before making a choice. This information may be costly, as it is difficult
to acquire more precise information, and there may be flexibility in tailoring its content to the context of the problem. To model this behavior, the rational inattention literature has, to a large extent, focused on 
\emph{posterior separable} costs of information, where the cost of an
experiment $\pi$ is given by 
\[
c(\pi)=\int \psi(p)\pi(d p).
\]

Posterior-separable costs allow the use of information design
tools to solve the rational inattention problem. To understand how this works, we first represent the indirect utility that the DM gets from a menu as a convex function $\phi$ of their belief about the state of the world, $p$. Then, the rational inattention problem can be written as
\[
\max_{\pi \in \Pi(p_0)}\int [\phi(p)-\psi(p)]\pi(d p),
\]
where $\Pi(p_0)$ is the set of distributions over posteriors consistent with the prior belief $p_0$. This problem is analogous to an information design problem, and its value function can be written as the concavification of the function $\phi-\psi$ evaluated at $p_0$ \citep{Aumann1995,Kamenica2011}. Another way to think of this concavification is by looking at the lowest value $A(p_0)$ among all affine functions $A$ satisfying $A\geq \phi-\psi$ or, geometrically, as the lowest hyperplane above the graph of $\phi-\psi$, which we call an ``optimal hyperplane" for the menu. A solution $\pi$ to the rational inattention problem must then have its support within the set of posterior beliefs such that $A(p)=\phi(p)-\psi(p)$.

In this paper, we show an axiomatic characterization of menu preferences coming from posterior-separable costs that is closely connected to the information-design formulation of the rational inattention problem. We build on the general cost characterization from \cite{de2017rationally}, adding two axioms. The first axiom, \emph{Independence
of Irrelevant Alternatives}, states that if the DM is indifferent 
between two menus $F,G$ as well as their intersection, 
then the
DM is also indifferent to their union. The idea is that the
preferences indicate that the options outside of the intersection
are not useful, and so the DM does not benefit from the added 
flexibility. Thus, they remain without benefit when considering
the union.

The second axiom, \emph{Ignorance Equivalence}\footnote{This axiom was called ``Linearity" in \cite{de2014axiomatic}.}, states that, for each menu,
there is an act that the DM is indifferent to, including under
union with the menu. This act provides an ``ignorance equivalent"
\citep{muller2023rational} to the menu. When this act is added to the menu, the agent is indifferent between acquiring information optimally for the menu or not acquiring any information and simply choosing the act. Thus, the extra flexibility afforded by adding the act does not alter the menu's value. As \cite{muller2023rational} explain, the ignorance equivalent 
can be viewed as an analogue of a certainty equivalent for menus.

Our characterization shows that these two axioms are equivalent to the existence of a posterior separable cost satisfying one extra property: for every menu, there is a unique optimal hyperplane. In principle, although two lowest hyperplanes must intersect at $p_0$, they could have different slopes. Thus, some posterior-separable costs do not satisfy this condition. We show that, to satisfy this unique hyperplane property, $\psi$ must satisfy a weak form of smoothness, which we call ``joint-directional differentiability". The set of such $\psi$ is dense  and contains the set of differentiable $\psi$, which includes all commonly used cost functions, such as entropy \citep{sims2003implications, de2014axiomatic, matvejka2015rational, caplin2022rationally}, log-%
likelihood \citep{pomatto2023cost}, residual variance \citep{ely2015suspense}, and neighborhood-based costs \citep{hebert2021neighborhood}. As we show 
in Section 6, when considering uniformly posterior-separable costs, 
our two axioms yield a differentiable cost representation.
\section{Model}
Let $\Omega$ denote a finite set of states of the world and let $X$ be a mixture space of consequences\footnote{For example, $X$ could be the set of lotteries
over a fixed set of prizes, or it could be a convex subset of some vector space.}. An \emph{act }is a function $f:\Omega\rightarrow X$. A finite set of
acts will be called a \emph{menu }and denoted by $F,G,H$ etc. The
set of all acts is denoted by $\mathcal{F}$ and the set of all menus by $\mathbb{F}$. A single act $f$ can also be seen as a singleton menu $\{f\}$; we usually omit the brackets if there is no chance for confusion. 

Mixtures of acts are defined pointwise:
given two acts $f,g$ and a scalar $\alpha\in\left[0,1\right]$, denote
by $\alpha f+\left(1-\alpha\right)g$ the act that in each state $\omega$
delivers the outcome $\alpha f\left(\omega\right)+\left(1-\alpha\right)g\left(\omega\right)$.
For $\alpha\in\left[0,1\right]$, the mixture of two menus is defined as
\[
\alpha F+\left(1-\alpha\right)G=\left\{ \alpha f+\left(1-\alpha\right)g:\:f\in F,g\in G\right\} .
\]
We can interpret $\alpha F+\left(1-\alpha\right)G$ as a lottery over
what menu the agent faces.

Given an arbitrary set $Z$ we let $\Delta(Z)$ denote the set of probability distributions over $Z$ with finite support. 

\subsection{Rationally inattentive preferences}

The primitive is a preference $\succsim$ defined over menus, which is interpreted according to the following
timeline. 

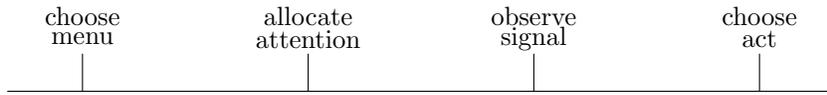
\begin{center}
\begin{tikzpicture}
\draw[->] (-1,0) -- (10,0);
\draw (0,0) -- (0,0.5);
\draw (3,0) -- (3,0.5);
\draw (6,0) -- (6,0.5);
\draw (9,0) -- (9,0.5);
\node at (0,1) {choose};
\node at (0,0.7) {menu};
\node at (3,1) {allocate};
\node at (3,0.7) {attention};
\node at (6,1) {observe};
\node at (6,0.7) {signal};
\node at (9,1) {choose};
\node at (9,0.7) {act};
\end{tikzpicture}
\end{center}
\begin{sloppypar}
Thus, the agent chooses among menus while aware that they will be able to obtain information before finally choosing an act. We consider an agent who is rationally inattentive, that is, whose preferences can be represented by
\[
V(\phi_F)=\max_{\pi \in \Pi(p_0)}\int \phi_F(p) \pi(dp)-c(\pi),
\]
where
\[
\Pi(p_0)=\left\{\pi\in \Delta(\Delta(\Omega))\Big|  \int p(\omega)\pi(dp)=p_0(\omega)\,\forall \omega\in \Omega\right\}
\]
represents distributions over posterior beliefs consistent with possible finite information structures, ${c:\Pi(p_0)\rightarrow \mathbb{R}\cup \{\infty\}}$ is the cost of information\footnote{The value of $\infty$ is assigned to those distributions over posteriors that should never be acquired, representing impossible information. This could alternatively be modeled as a restriction on the domain of the cost function.}, and
\[
\phi_F(p)=\max_{f\in F}\sum_{\omega} u(f(\omega))p(\omega),
\]
 where $u:X\rightarrow \mathbb{R}$ is the utility function over consequences. We assume that the image of $u$ is $\mathbb{R}$. 
\end{sloppypar}

 Note that, in the representation above, the consequence associated with an act in a given state only matters insofar as it affects utility. Thus, we may consider, as shorthand, instead of acts $f:\Omega\rightarrow X$, \emph{utility acts} given by $u\circ f:\Omega\rightarrow \mathbb{R}$. We may also work with \emph{utility menus}---finite sets of utility acts. For instance, when we refer to the menu $F=\{0\}$, we mean a utility menu that contains a single utility act $0\in \mathbb{R}^\Omega$ or, equivalently, any menu that has a single act giving utility zero in every state. 

The following result is proved in \cite{de2017rationally}:

\begin{proposition} [\cite{de2017rationally}, Theorems 1 \& 2]\label{prop:de2017rationally}
    Let $\succsim$ be a rationally inattentive preference. Then $\succsim$ has a representation $(u,p_0,c)$ where the cost function $c:\Pi(p_0)\rightarrow \mathbb{R}\cup \{\infty\}$ is \emph{canonical}, i.e. it satisfies
    \begin{description}
    \item[Groundedness:] $c(\delta_{p_0})=0$,
    \item[Convexity:] $c$ is a convex function,
    \item[Lower-semicontinuity:] $c$ is lower-semicontinuous in the weak$^*$ topology over $\Pi(p_0)$, and
    \item[Blackwell Monotonicity:] $c$ is increasing in the Blackwell order.
\end{description}
Moreover, fixing $u$, this cost function is unique and can be recovered from the functional $V$ by the formula
\[
c(\pi)=\sup_{F\in\mathbb{F}} \int \phi_F(p) \pi(dp) - V(\phi_F).
\]  
\end{proposition}

The following is a useful characterization of canonical costs (slightly modified from \cite{denti2021experimental}, Lemma 6):
\begin{proposition}\label{prop:dmr2021}
    A cost function $c:\Pi(p_0)\rightarrow \mathbb{R}\cup \{\infty\}$ is canonical if and only if it can be written as 
    \begin{equation}
\label{gencostrep}
    c(\pi)=\sup_{\psi\in\Psi} \int \psi \pi(dp)
\end{equation}
where $\Psi$ is a set of convex functions $\psi:\Delta(\Omega)\rightarrow \mathbb{R}\cup \{\infty\}$, such that
\begin{enumerate}
    \item $\max_{\psi\in \Psi} \psi(p_0)=0$, and
    \item $\Psi$ is \emph{minimal}\,---\,there is no 
    $\hat{\Psi}\subset \Psi$ such that $\sup_{\psi\in\Psi} \int \psi d\pi=\sup_{\psi\in\hat{\Psi}} \int \psi d\pi$ and 
    $\vert \hat{\Psi}\vert <\vert \Psi\vert$.\footnote{We define $|\Psi|\in \mathbb{N}\cup \{\infty\}$ as the cardinality of the set $\Psi$, without making distinction between countable and uncountable infinities.}
\end{enumerate}
\end{proposition}
\begin{proof}
See \Cref{app:proof_dmr2021}.    
\end{proof}

The main result in \cite{de2017rationally} is an axiomatic characterization of rationally inattentive preferences. The following properties will be useful for us:

\begin{enumerate}
    \item If $F\subseteq G$ then $V(\phi_F)\leq V(\phi_G)$ (``preference for flexibility");
    \item For any menu $F$ and act $h$, 
    \[
    V(\phi_F+\phi_h)=V(\phi_F)+\phi_h(p_0).
    \]
\end{enumerate}

\subsection{Posterior-separable cost}

Our goal in this paper is to understand a more specific class of costs of information, which we introduce now. 

\begin{definition}
    The cost of information $c:\Pi(p_0)\rightarrow \mathbb{R}\cup \{\infty\}$ is said to be \emph{posterior separable} if there exists a measurable, bounded from below,%
    \footnote{Notably, $\psi$ can take on values of $+\infty$.}
    and lower-semicontinuous function $\psi:\Delta(\Omega)\rightarrow \mathbb{R}\cup \{\infty\}$ such that, for all $\pi\in\Pi(p_0)$,
    \[
    c(\pi)=\int \psi(p)\pi(dp).
    \]
    In this case, we say that $c$ is represented by $\psi$, or that $\psi$ is a representation of $c$, and call $\psi$ a \emph{measure of uncertainty}.
\end{definition}

In the remainder of this subsection, we propose a class of functions $\psi$
that have certain properties that are convenient to use and without loss of generality. 
Given a measure of uncertainty $\psi$, let 
\[
\dom \psi=\left\{p\in \Delta(\Omega)\:|\:\psi(p)<\infty\right\}
\]
be the \emph{effective domain} of $\psi$.
\begin{definition}\label{def:canonical}
    A measure of uncertainty
    $\psi:\Delta(\Omega)\rightarrow\mathbb{R}\cup\{\infty\}$ is \emph{canonical} if:
    \begin{enumerate}
        \item \label{psi_convex}$\psi$ is convex;
        \item \label{psi_lsc}$\psi$ is lower-semicontinuous;
        \item \label{psi_grounded} $\psi(p_0)=0$;
        \item \label{psi_positive} $\psi\geq 0$.
        \item \label{ri_prior} $p_0\in \ri(\dom\psi)$;
    \end{enumerate}
\end{definition}

 Before stating our formal result, we discuss some intuition for why these properties can be assumed without loss of generality. Convexity follows from Blackwell monotonicity of $c$. Lower-semicontinuity ensures that $V(\phi_F)$ is well-defined. Property \ref{psi_grounded} follows from groundedness of $c$. Property \ref{psi_positive} is a normalization that can be achieved by noting that adding an affine function to $\psi$ that is zero at $p_0$ does not affect the cost of information. These first four properties are well known and often assumed whenever posterior-separable costs are used. To our knowledge, Property \ref{ri_prior} is new. 
 \begin{wrapfigure}{r}{0.45\textwidth}
    \centering
    \begin{tikzpicture}
  \draw (0,-2) -- (0,2);
  \filldraw[fill=domainblue, draw=black]
    [xscale=1.5, yscale=2]
      (0,-1) arc[start angle=-90, end angle=90, radius=1]
      -- (0,-1) -- cycle;
    \node at (0,0) [circle,fill, inner sep=1pt](prior){};
    \node[left =1pt of prior]{$p_0$};
    \node at (.8,0) {dom $\psi$};
\end{tikzpicture}
    \caption{$p_0$ not in relative interior}
    \label{fig:not_ri}
\end{wrapfigure}
 To see why it can be assumed without loss of generality, suppose that $p_0$ is not in the relative interior, as in \Cref{fig:not_ri}. There, $p_0$ lies in the vertical line that describes the left of the boundary of $\dom \psi$. Any $\pi$ that puts positive probability on the right side of the line must also put positive probability on the left side of the line, where $\psi=\infty$, hence $c(\pi)=\infty$. Thus, if $c(\pi)<\infty$, it must be that the support of $\pi$ is contained within the vertical line. This means that nothing is lost by redefining the effective domain of $\psi$ to be just the vertical line.

Below, we state the formal result, which also includes a sufficient condition for uniqueness. 

\begin{proposition}\label{prop:canonps}
    Let the information cost function $c$ be canonical, as defined in \Cref{prop:de2017rationally}. Suppose that $c$ is posterior separable and represented by $\psi$. Then there exists a canonical measure of uncertainty $\tilde{\psi}$ such that $c$ is represented by $\tilde{\psi}$. Moreover,  if $\psi$ is convex and differentiable in the directions of its effective domain at $p_0$, then $\tilde{\psi}$ is unique.
\end{proposition}
\begin{proof}
    See \Cref{app:proof_of_canonps}.\footnote{The proof uses a construction shown in \cref{sec:dimension}, so the reader is advised to read that section before reading the proof.}
\end{proof}
From here onward, unless specified otherwise, we restrict our focus to
posterior-separable costs that have a canonical measure of uncertainty.

\subsection{Concavification}
If the cost of information is posterior separable, we can write
\[
V(\phi_F)=\max_{\pi\in\Pi(p_0)}\int N_F(p)\pi(dp)
\]
where $N_F=\phi_F-\psi$ is the net utility. This parallels the 
objective in Bayesian persuasion, where the sender's objective is 
to find the optimal distribution with respect to the integrand. The optimum is found by taking the \emph{concavification} of $N_F$,
i.e.
\[
V(\phi_F)=\cav(N_F)(p_0)=\inf\{\zeta(p_0):\,\zeta\geq N_F,\,\zeta \mbox{ concave} \}.
\]
One can then use the techniques of finding the optimum from Bayesian
persuasion, as found in \cite{Aumann1995} and \cite{Kamenica2011}.

\subsection{Dimension of effective domain}\label{sec:dimension}

As will become clear as we develop the main result, it will be useful to 
embed $\dom(\psi)$ in a space that has the same dimension, so the prior becomes an interior point in this space. To this end, 
let $\aff(\dom(\psi))$ be the affine hull of 
$\dom(\psi)$ (the smallest affine set containing it) and let $M$ be the dimension of this space. By \cite{Rockafellar+1970}, 
Theorem 1.6, there exists a 
bijective affine transformation 
\begin{equation}
\label{domnorm}
    T:\aff(\dom(\psi))\rightarrow\mathbb{R}^M
\end{equation}

For instance, in the case of full domain, where $\dom(\psi)=\Delta(\Omega)\subset \mathbb{R}^\Omega$,
one can define such a $T$ as mapping to $\mathbb{R}^{\vert \Omega\vert-1}$
by setting $T_i(p)=p(\omega_i)$ for $i\in \{1,...,\vert \Omega\vert-1\}$,
implicitly determining the probability of the remaining state 
$\omega_{\vert \Omega\vert}$ as $1-\sum_{i<\vert\Omega\vert}p(\omega_i)$.

Given any $\psi:\Delta(\Omega)\rightarrow \mathbb{R}\cup\{\infty\}$, we define $T^*\psi:\mathbb{R}^M\rightarrow \mathbb{R}\cup\{\infty\}$ by 
\begin{equation}
\label{psinorm}
    T^* \psi(y)=\begin{cases}
    \psi(T^{-1}(y)), & y\in T(\dom(\psi))\\
    \infty, & \mbox{otherwise}
\end{cases}
\end{equation}
Notice that, since $T$ is a linear transformation, 
the properties of a canonical $\psi$ are slightly modified by $T^*\psi$ as follows:
\begin{enumerate}
    \item $T^*\psi$ is convex;
    \item $T(p_0)\in \mbox{int}(T(\dom(\psi))$;
    \item $T^*\psi(T(p_0))=0$;
    \item $T^*\psi\geq 0$.
\end{enumerate}
Notice, in particular, that since $\dim(T(\dom(\psi)))=M$, $T(p_0)$ is
now in the \emph{interior} of $T(\dom(\psi))$, not just the relative 
interior. To visualize this, Figure 2 shows such a transformation,
where $\dom(\psi)$ lies along a line where the prior is collinear with
$\delta_{\omega_3}$; thus, learning anything about the relative likelihood  between $\omega_1$ and $\omega_2$ is infinitely costly. As such, the potential learning of the DM is 
one-dimensional. Therefore, one can translate the information 
acquisition problem to a one-dimensional effective domain, as illustrated by $T$.

\begin{figure}[t!]
\centering

\begin{tikzpicture}[
    scale=1.2,
    >=Stealth,
    point/.style={circle, inner sep=1.5pt, fill=black},
    redpoint/.style={circle, inner sep=1.5pt, fill=red},
    axis/.style={->, thick, blue!80!black}
]

    
    \coordinate (w1) at (0,0);
    \coordinate (w2) at (4,0);
    \coordinate (w3) at (2,3.5);
    \coordinate (A) at (2,0.3);
    \coordinate (B) at (2,1.9);

    \draw[thick] (w1) -- (w2) -- (w3) -- cycle;

    \node[below left] at (w1) {$\omega_1$};
    \node[below right] at (w2) {$\omega_2$};
    \node[above left] at (w3) {$\omega_3$};

    \draw[ultra thick, domainblue] (A) -- (B);

    \draw[thick, blue] (1.95,1.9) -- (2.05, 1.9) node[left] {\footnotesize $T^{-1}(b)$};
    \draw[thick, blue] (1.95,0.3) -- (2.05, 0.3) node[left] {\footnotesize $T^{-1}(a)$};
    
    \coordinate (p0) at (2, 1.2); 
    \node[redpoint] at (p0) {};
    \node[left, red, align=left, font=\scriptsize, xshift=6pt] at (p0) 
    {$p_0$ \hspace{2 mm}};
    
    \node[red, font=\large] at (2,1.6) {\rotatebox{90}{)}};
    \node[red, font=\large] at (2,0.8) {\rotatebox{90}{(}};

    \node[right, blue] at (2.05, 0.5) {\scriptsize $dom(\psi)$};
    
    \draw[thick, domainblue, dashed] (A) -- (2,-1);
    \node[right, blue, font=\scriptsize] at (2, -0.6) {$\aff(\dom(\psi))$};
    \draw[thick, domainblue, dashed] (B) -- (2,4.5);


    \begin{scope}[shift={(7, 1)}]
        
        \draw[->, thick, purple!80!black] (-2, 0) -- (3, 0) node[right] {$\mathbb{R}$};
        
        \draw[ultra thick, domainblue] (-1, 0) -- (1, 0);
        
        \draw[thick, blue] (-1, 0.1) -- (-1, -0.1) node[below] {$a$};
        \draw[thick, blue] (1, 0.1) -- (1, -0.1) node[below] {$b$};
        
        \node[below, blue] at (0, -0.35) {$T(\dom(\psi))$};
        
        \coordinate (Tp0) at (0,0);
        \node[redpoint] at (Tp0) {};
        \node[above, red, font=\scriptsize, yshift=3pt] at (0, 0) 
        {$T(p_0)$};
        
        \node[red, font=\large] at (-0.4, 0) {(};
        \node[red, font=\large] at (0.4, 0) {)};

    \end{scope}

    \draw[->, thick, purple!80!black, bend left=45] (2.1, 1.7) to node[midway, above] {$T$} (6.5, 1.2);

\end{tikzpicture}

\caption{Linear transformation $T$ for dimension matching}
\end{figure}

To economize on notation, we write $\bar{\psi}\coloneqq T^*\psi$, $\bar{p}\coloneqq T(p)$, and $Y\coloneqq T(\dom(\psi))$ in the sequel.

\subsection{Unique Hyperplane Property}

In this section, we present an alternative description of concavification, equivalent to that presented in Section 2.3. This relies on the supporting hyperplane of the concave function $\cav(N_F)$, so that all values of the function lie below this hyperplane.%
\footnote{This is analogous to the ``Lagrangian lemma" of \cite{caplin2022rationally}.}
We use this representation of the concavification to develop key properties
of our cost function representation.

Let $T$ be as in \eqref{domnorm}. For any hyperplane 
$\mathcal{H}\in \mathbb{R}^{M+1}$, let $\lambda$ be its normal vector. 
Define $\bar{N}_F:\mathbb{R}^M\rightarrow \mathbb{R}$ by 
\begin{equation}
    \bar{N}_F(y)=\begin{cases}
    N_F(T^{-1}(y)), & y\in Y\\
    -\infty, & \mbox{otherwise}
    \end{cases}
\end{equation}
An equivalent formula for the concavification for a given prior $\bar{p}_0$ is\footnote{This formula follows from \cite{Rockafellar+1970}, Theorem 18.8, which states 
that any closed convex set in a Euclidean space is the intersection of 
the closed half-spaces tangent to it. Since the epigraph $\epi(-\cav(\bar{N}_F))=\{(\bar{p},t)|t\geq -\cav(\bar{N}_F)(\bar{p})\}$ is 
closed and convex, there exist such tangent hyperplanes as described in 
(\ref{eqn:min_concavification}) for $\bar{p}=\bar{p}_0$, so the minimum is achieved.}

\begin{equation}\label{eqn:min_concavification}
    \cav(\bar{N}_F)(\bar{p}_0)=\min_{\substack{\vphantom{.}\\ \mathllap{\lambda}\in \mathrlap{\mathbb{R}^M }\\ \mathllap{k}\in \mathrlap{\mathbb{R}}}}\{\lambda\cdot \bar{p}_0+k:\,\lambda\cdot y+k\geq \bar{N}_F(y),\,\forall y\in Y \}.
\end{equation}
A pair $(\lambda,k)$ is a solution to the minimization problem above if $\lambda\cdot y+k\geq \bar{N}_F(y)$ for all $y \in Y$ and $\lambda\cdot \bar{p}_0+k=\cav(\bar{N}_F)(\bar{p}_0)=V(\phi_F)$. Thus, we can solve for $k$ in this expression and denote the set of solutions to \eqref{eqn:min_concavification} by
\[
\Lambda_F=\left\{ \lambda\in \mathbb{R}^{M}|\lambda\cdot (y-\bar{p}_0)+V(\phi_F)\geqslant 
\bar{N}_F(y), \, \forall y\in Y
\right\}
\]
By \cite{Rockafellar+1970}, Theorem 23.2, $\Lambda_F$ is closed, convex, 
and non-empty. Geometrically, the elements of $\Lambda_F$ are the normal vectors of hyperplanes that are tangent to the graph of $\bar{N}_F$ (see \cref{fig:IIA}). When $\bar{N}_F$ is well-behaved 
(in a sense that will be made precise shortly), this set is a singleton, motivating the following definition:

\begin{definition}
    $\psi$ satisfies the \emph{unique hyperplane property (UHP)} if, for all 
    menus $F$, $\Lambda_F$ is a singleton.
\end{definition}

\subsection{Joint-Directional Differentiability}

As will become clear later, the unique hyperplane property is related to a notion of differentiability that we now discuss. 

\begin{definition}
    Let $\bar{\psi}:Y\rightarrow \mathbb R \cup \{\infty\}$ be as defined
    in Section 2.4. The \emph{subdifferential} of $\bar{\psi}$ at $y\in Y$ is
    \[
      \partial \bar{\psi}(y)=\{\lambda \in \mathbb{R}^M:\, \lambda\cdot(q-y)\leq \bar{\psi}(q)-\bar{\psi}(y),\forall q\in Y\}.
    \]
\end{definition}

    A convex function is differentiable if and only if its
    subdifferential is a singleton (\cite{Rockafellar+1970}, Theorem 25.1). To compare the subdifferentials at various points when they are not singletons, we introduce the following property
    of $\psi$.
    
\begin{definition} \label{def:NDISD}
    The function $\bar{\psi}$ is \emph{non-differentiable
        in the direction $\delta \in \mathbb{R}^M\setminus \{0\}$ at $\{\bar{p}_i\}_{i=1}^K$} if there exist $\lambda_i\in \partial\bar{\psi}(\bar{p}_i)$ such that:
        \begin{enumerate}
            \item $\lambda_i+\delta\in \partial\bar{\psi}(\bar{p}_i),\forall i\in\{1,...,K\}$ and 
            \item $\delta\cdot (\bar{p}_i-\bar{p}_0)=0,\forall i\in\{1,...,K\}$.
        \end{enumerate}
        We say that $\bar{\psi}$ is \emph{non-differentiable
        in the same direction (NDISD) at $\{\bar{p}_i\}_{i=1}^K$} if there exists a $\delta \in \mathbb{R}^M\setminus \{0\}$ such that $\bar{\psi}$ is non-differentiable
        in the direction $\delta$ at $\{\bar{p}_i\}_{i=1}^K$.
\end{definition}
Geometrically, we can think of an element of the subdifferential as the slope of a hyperplane that is tangent to the graph of $\bar{\psi}$ at the point $\bar{p}_i$. When there is more than one such tangent hyperplane, it means that the function $\bar{\psi}$ has a kink at that point, and we can think of the difference between the two slopes ($\delta$) as a direction of that kink, in the sense that one can wobble the hyperplane by adding $\delta$ and remain tangent. The condition above states that at all the points $\{\bar{p}_i\}_{i=1}^K$ there is a kink and that they all share a wobbling direction $\delta$.  Thus, the direction $\delta$ is orthogonal to the convex hull of $\{\bar{p}_i-\bar{p}_0\}_{i=1}^K$. We illustrate this in Figure 3.

\begin{figure}[t!]
\centering
\tdplotsetmaincoords{70}{120}

\begin{tikzpicture}[
    tdplot_main_coords,
    scale=2,
    >=Stealth,
    font=\footnotesize
]

    \coordinate (RidgeStart) at (0, 0, 1.5);
    \coordinate (RidgeEnd)   at (0, 6, 1.5);
    
    \coordinate (BaseL_Start) at (-1.5, 0, 0);
    \coordinate (BaseL_End)   at (-1.5, 6, 0);
    \coordinate (BaseR_Start) at (1.5, 0, 0);
    \coordinate (BaseR_End)   at (1.5, 6, 0);
    
    \coordinate (P1) at (0, 1.5, 1.5); 
    \coordinate (P0) at (0, 3.0, 1.5); 
    \coordinate (P2) at (0, 4.5, 1.5); 

    \coordinate (Y1) at (0, 1.5, 0);
    \coordinate (Y0) at (0, 3.0, 0);
    \coordinate (Y2) at (0, 4.5, 0);

    
    \fill[red!70!black, opacity=0.8] (RidgeStart) -- (BaseL_Start) -- (BaseL_End) -- (RidgeEnd) -- cycle;
    
    \fill[red!50!white, opacity=0.9] (RidgeStart) -- (BaseR_Start) -- (BaseR_End) -- (RidgeEnd) -- cycle;
    
    \draw[red!30!black, thin] (RidgeStart) -- (BaseL_Start) -- (BaseL_End) -- (RidgeEnd);
    \draw[red!30!black, thin] (RidgeStart) -- (BaseR_Start) -- (BaseR_End) -- (RidgeEnd);
    \draw[gray, thin] (BaseL_Start) -- (BaseR_Start);

    \draw[->, gray, thin] (0,0,0) -- (2,0,0);   
    \draw[->, gray, thin] (0,0,0) -- (0,6.5,0); 
    \draw[->, gray, thin] (0,0,0) -- (0,0,2.5); 

    
    \fill[black] (Y1) circle (1pt) node[right] {$p_1$};
    \fill[black] (Y0) circle (1pt) node[right] {$p_0$};
    \fill[black] (Y2) circle (1pt) node[right] {$p_2$};

    \draw[dashed, blue, thick] (P1) -- (Y1);
    \draw[dashed, blue, thick] (P0) -- (Y0);
    \draw[dashed, blue, thick] (P2) -- (Y2);

    \fill[red!50!black] (P1) circle (1pt);
    \fill[red!50!black] (P0) circle (1pt);
    \fill[red!50!black] (P2) circle (1pt);

    \draw[->, red!50!black, very thick] ($(P0)+(0,0,0.0)$) -- ($(P1)+(0,0,0.0)$) 
        node[midway, above, yshift=2pt, font=\scriptsize] {$p_1 - p_0$};
        
    \draw[->, red!50!black, very thick] ($(P0)+(0,0,0.0)$) -- ($(P2)+(0,0,0.0)$) 
        node[midway, above, yshift=2pt, font=\scriptsize] {$p_2 - p_0$};

    
    \def\len{1.0}
    \def\lx{-0.7} \def\lz{0.8} 
    \def\rx{0.7}  \def\rz{0.8} 
    
    \coordinate (L1_P1) at ($(P1) + (\rx, 0, \rz)$);
    \coordinate (L2_P1) at ($(P1) + (\lx, 0, \lz)$);

    \draw[->, violet, thick] (P1) -- (L1_P1) node[above left] {$\lambda_1$};
    \draw[->, violet, thick] (P1) -- (L2_P1) node[above left] {$\lambda_2$};

    \coordinate (L1_P2) at ($(P2) + (\rx, 0, \rz)$);
    \coordinate (L2_P2) at ($(P2) + (\lx, 0, \lz)$);
    
    \draw[->, violet, thick] (P2) -- (L1_P2) node[above left] {$\lambda_1$};
    \draw[->, violet, thick] (P2) -- (L2_P2) node[above left] {$\lambda_2$};

    
    \draw[->, green!60!black, ultra thick] (L2_P1) -- (L1_P1) 
        node[midway, above, fill=white, inner sep=1pt] {$\delta$};
    
    \draw[->, green!60!black, ultra thick] (L2_P2) -- (L1_P2) 
        node[midway, above, fill=white, inner sep=1pt] {$\delta$};

\end{tikzpicture}

\caption{Non-differentiability in direction $\delta$}
\end{figure}

This definition will be useful in constructing menus that take advantage
of the directions of non-differentiability to violate our axioms. In
particular, it enables us to define cost functions that are 
``sufficiently smooth."

\begin{definition}
    Given prior $p_0$, the function $\psi$ satisfies \emph{joint-directional 
    differentiability (JDD)} if there do not exist $\delta\in\mathbb{R}^M\setminus \{0\}$ 
    and points $\{p_i\}_{i=1}^K$, with $p_0\in co({p_1,\ldots,p_K})$,
     such that $\bar{\psi}$ is non-differentiable in the 
    direction $\delta$ at $\{\bar{p}_i\}_{i=1}^K$.
\end{definition}

Notice that under this definition, a function that is differentiable everywhere
satisfies JDD, as there are no points at all at which $\bar{\psi}$ is 
non-differentiable. In addition, in the case of a binary state space, JDD is 
equivalent to differentiability at the prior $p_0$, as condition (2) of \Cref{def:NDISD} precludes and $\delta\neq 0$ for $\bar{p}_i\neq \bar{p}_0$ when $M=1$.

\noindent\textbf{Remark:} Whether $\psi$ satisfies JDD does not 
depend on the choice of $T$. Suppose $\hat{T}$ is another choice. Because $T$ and $\hat{T}$ have the same dimension in the domain and codomain, there must be an invertible linear transformation $A:\mathbb{R}^M\rightarrow \mathbb{R}^M$ such that $\hat{T}=A\circ T$. Letting
$\mathbf{A}$ be the matrix associated with $A$, if $\lambda\in \partial \bar{\psi}(x)$, then 
$(\mathbf{A}^{-1})^\intercal\lambda \in \partial (A\circ T)^* \psi$, where
$(\mathbf{A}^{-1})^\intercal$ is the transpose of the inverse of $\mathbf{A}$.

\section{Main Theorem}

\subsection{Axioms}

Our main result relies on the following axioms.

\begin{axiom:iia}
\phantomsection
\label{axiom:iia}
    If $F\sim F\cap G \sim G$ then $F\sim F \cup G$.
\end{axiom:iia}

For any menus $F$ and $G$, we always have $F\cup G\succsim F,G\succsim F\cap G$, because flexibility is never harmful. When $F\sim F\cap G$, we may say that the elements of $F$ that do not belong to $F\cap G$ are irrelevant: the DM can achieve the
same payoff even by ignoring these additional options. The axiom states that these irrelevant options remain irrelevant when combined with other irrelevant options. 

\begin{axiom:ie}%
\label{axiom:ie}
    For every menu F, there exists an act $h$ such that $h\sim F\sim F\cup h$.
\end{axiom:ie}

When faced with a singleton menu $h$, it is always optimal for the agent to acquire no information. The act $h$ in this axiom is just as good as $F$, yet it adds irrelevant flexibility to $F$. Thus, the act $h$ can be thought of as an ignorance equivalent---a version of the menu $F$ that requires no information to be acquired. 

First appearing in \cite{de2014axiomatic}, this property also appears in 
\cite{muller2023rational}, who coined the term ``ignorance equivalent'' to refer to $h$. The term comes from an analogy with the role of a certainty equivalent  in the context of decisions
under risk. In decisions under risk, a risk-neutral principal can extract the most surplus by offering the certainty equivalent (i.e. fully 
insuring). Analogously, in rational inattention, the principal could offer the ignorance equivalent to keep the agent's surplus at the same level, while saving
all information acquisition costs; thus the additional surplus would entirely go to the principal. The ignorance equivalent also serves as a tool to see which actions might 
ever be chosen in an expanded menu: when comparing $F$ and $F\cup\{g\}$, it is 
sufficient to know that $h$, the ignorance equivalent, dominates $g$ (i.e. is better in every state) in order to 
conclude that $g$ will never be chosen from $F\cup\{g\}$, and therefore 
$F\sim F\cup\{g\}$.

\subsection{Main theorem}

We can now state our main result:

\begin{theorem}
    \label{possep}
    The following statements are equivalent:
    \begin{enumerate}
        \item \label{item:iia_and_ie}$\succsim$ is a rationally inattentive preference satisfying  \iia and \ie;
        \item \label{item:jdd}$\succsim$ has a posterior-separable representation with a canonical measure of uncertainty $\psi$ satisfying joint-directional differentiability.
        \item \label{item:uhp}$\succsim$ has a posterior-separable representation with a canonical measure of uncertainty $\psi$ satisfying the unique hyperplane property.
    \end{enumerate}
    Moreover, fixing the utility function $u$, the canonical measure of uncertainty $\psi$ that represents $\succsim$ in (2) and (3) is unique.
\end{theorem}

\subsection{Example}

To see how the three statements in Theorem \ref{possep} are related, we 
present the following example. In particular, we highlight the role that 
the unique hyperplane property has in our analysis, showing that, as the 
unique hyperplane property is violated, this leads to a violation of \iia. 

Consider a setup with two states, in which the prior $p_0=0.5$ and 
the  cost of information is given (\Cref{fig:IIA}a) by 
\[
\psi(p)=\vert p-0.5\vert +(p-0.5)^2
\]
As $\psi$ is non-differentiable at the prior, it violates joint directional
differentiability. Indeed, with the menu $F=\{0\}$, there are multiple optimal 
hyperplanes, all of which yield $V(F)=0$. Now suppose we consider two 
new actions $\{a,b\}$, with respective payoffs $\phi_{\{a\}}(p)=2.5p-\frac{21}{16}$
and $\phi_{\{b\}}(p)=-2.5p+\frac{19}{16}$. One easily verifies that the DM is indifferent 
between the menus $\{0\}$, $\{0,a\}$, and $\{0,b\}$, as seen by the fact 
that all three generate respective optimal hyperplanes $\lambda$
such that the value of $\lambda$ at $p_0=0.5$ is $0$ (\Cref{fig:IIA}b and \Cref{fig:IIA}c). 
However, if we take the union $\{0,a,b\}$, the optimal hyperplane changes: 
it now becomes optimal to choose posteriors $p\in\{0,1\}$ (\Cref{fig:IIA}d), 
yielding a value of the menu of $\frac{3}{8}$. As $a$ and $b$ are therefore
only relevant when we take the union of the menus, but not when we add
them individually to $0$, \iia is not satisfied.

\begin{figure}[t!]
\centering

\begin{subfigure}[t]{0.4\textwidth}
\resizebox{60 mm}{45 mm}{
\begin{tikzpicture}
    \begin{axis}[
        samples=500,
        ytick = {0},
        yticklabels={$0$},
        xtick={0,1},
        ymajorticks=false,
        ymax=1,
        ymin=-0.5,
        xmin=-0.1,
        xmax=1.1,
        axis on top=false,
        axis x line = middle,
        axis y line = middle,
        axis line style={black},
        ylabel={$u$},
        xlabel={$p$},
    ]     
    \addplot[red,samples=200,line width = 2][domain=0:0.5] {0.5-x+(x-0.5)^2};
    \addplot[red,samples=200,line width = 2][domain=0.5:1] {x-0.5+(x-0.5)^2};
    \end{axis}
\end{tikzpicture}
}
\caption{$\psi$ that is NDISD}
\end{subfigure}
\hspace{1 mm}
\begin{subfigure}[t]{0.4\textwidth}
\resizebox{60 mm}{45 mm}{
\begin{tikzpicture}
    \begin{axis}[
        samples=500,
        ytick = {0},
        yticklabels={$0$},
        xtick={0,1},
        ymajorticks=false,
        ymax=1,
        ymin=-0.5,
        xmin=-0.1,
        xmax=1.1,
        axis on top=false,
        axis x line = middle,
        axis y line = middle,
        axis line style={black},
        ylabel={$u$},
        xlabel={$p$},
    ]     
    \addplot[red,samples=200,line width = 2][domain=0:0.5] {x-0.5-(x-0.5)^2};
    \addplot[red,samples=200,line width = 2][domain=0.5:1] {-x+0.5-(x-0.5)^2};
    \node[fill, red, circle, inner sep=1.5pt] 	at (axis cs:0.5,0){};
    \addplot[orange,samples=200,line width = 2][domain=0:0.5] {2.5*x-1.3125+x-0.5-(x-0.5)^2};
    \addplot[orange,samples=200,line width = 2][domain=0.5:1] {2.5*x-1.3125-x+0.5-(x-0.5)^2};
    \addplot[orange, dashed, samples=200, line width = 1][domain=0:1] {x-0.5};
    \node[fill, orange, circle, inner sep=1.5pt] 	at (axis cs:0.75,0.25){};
    \node[orange, below, yshift=-4pt] at (axis cs:0.75,0.25) {$\bar{N}_{\{a\}}$};
    \node[red, right] at (axis cs:0.8,-0.35) {$\bar{N}_{\{0\}}$};
    \end{axis}
\end{tikzpicture}
}
\caption{Optimal hyperplanes of menu $\{0,a\}$}
\end{subfigure}

\vspace{5 mm} 

\begin{subfigure}[t]{0.4\textwidth}
\resizebox{60 mm}{45 mm}{
\begin{tikzpicture}
    \begin{axis}[
        samples=500,
        ytick = {0},
        yticklabels={$0$},
        xtick={0,1},
        ymajorticks=false,
        ymax=1,
        ymin=-0.5,
        xmin=-0.1,
        xmax=1.1,
        axis on top=false,
        axis x line = middle,
        axis y line = middle,
        axis line style={black},
        ylabel={$u$},
        xlabel={$p$},
    ]     
    \addplot[red,samples=200,line width = 2][domain=0:0.5] {x-0.5-(x-0.5)^2};
    \addplot[red,samples=200,line width = 2][domain=0.5:1] {-x+0.5-(x-0.5)^2};
    \node[fill, red, circle, inner sep=1.5pt] 	at (axis cs:0.5,0){};
    \addplot[violet,samples=200,line width = 2][domain=0:0.5] {-2.5*x+1.1875+x-0.5-(x-0.5)^2};
    \addplot[violet,samples=200,line width = 2][domain=0.5:1] {-2.5*x+1.1875-x+0.5-(x-0.5)^2};
    \addplot[violet, dashed, samples=200, line width = 1][domain=0:1] {0.5-x};
    \node[fill, violet, circle, inner sep=1.5pt] 	at (axis cs:0.25,0.25){};
    \node[violet, below, yshift=-4pt] at (axis cs:0.25,0.25) {$\bar{N}_{\{b\}}$};
    \node[red, left] at (axis cs:0.8,-0.35) {$\bar{N}_{\{0\}}$};
    \end{axis}
\end{tikzpicture}
}
\caption{Optimal hyperplanes of menu $\{0,b\}$}
\end{subfigure}
\hspace{1 mm}
\begin{subfigure}[t]{0.4\textwidth}
\resizebox{60 mm}{45 mm}{
\begin{tikzpicture}
\begin{axis}[
        samples=500,
        ytick = {0},
        yticklabels={$0$},
        xtick={0,1},
        ymajorticks=false,
        ymax=1,
        ymin=-0.5,
        xmin=-0.1,
        xmax=1.1,
        axis on top=false,
        axis x line = middle,
        axis y line = middle,
        axis line style={black},
        ylabel={$u$},
        xlabel={$p$},
    ]     
    \addplot[red,samples=200,line width = 2][domain=0:0.5] {x-0.5-(x-0.5)^2};
    \addplot[red,samples=200,line width = 2][domain=0.5:1] {-x+0.5-(x-0.5)^2};
    \addplot[orange,samples=200,line width = 2][domain=0:0.5] {2.5*x-1.3125+x-0.5-(x-0.5)^2};
    \addplot[orange,samples=200,line width = 2][domain=0.5:1] {2.5*x-1.3125-x+0.5-(x-0.5)^2};
    \addplot[violet,samples=200,line width = 2][domain=0:0.5] {-2.5*x+1.1875+x-0.5-(x-0.5)^2};
    \addplot[violet,samples=200,line width = 2][domain=0.5:1] {-2.5*x+1.1875-x+0.5-(x-0.5)^2};
    \addplot[brown, dashed, samples=200, line width = 1][domain=0:1] {0.4375};
    \node[fill, brown, circle, inner sep=1.5pt] 	at (axis cs:0,0.4375){};
    \node[fill, brown, circle, inner sep=1.5pt] 	at (axis cs:1,0.4375){};
    \node[violet, below, yshift=-4pt] at (axis cs:0.25,0.25) {$\bar{N}_{\{a\}}$};
    \node[orange, below, yshift=-4pt] at (axis cs:0.75,0.25) {$\bar{N}_{\{a\}}$};
    \node[red, right] at (axis cs:0.8,-0.35) {$\bar{N}_{\{0\}}$};
    \end{axis}
\end{tikzpicture}
}
\caption{Improvement for taking union $\{0,a,b\}$}
\end{subfigure}

\caption{Posterior-separable cost function violating IIA}
\label{fig:IIA}
\end{figure}

\section{Proof}

In this section, we present a proof of our main result. To do so, we build 
on the connections that we illustrated in the example in Section 3.3 
between our two axioms and cost functions $\psi$ that satisfy UHP 
(alternatively, JDD). For a high-level discussion of the role that each
axiom plays in restricting the cost of information, see Section 5.

The proof goes as follows. \Cref{proof:posterior_separable} shows that, together, \iia and \ie imply that $\succsim$ has a posterior-separable representation. From this point on, all sections assume a posterior separable representation. \Cref{proof:UHP_implies_JDD} shows that the unique hyperplane property implies joint-direction differentiability and \Cref{proof:JDD_implies_UHP} shows that joint-directional differentiability implies the unique hyperplane property, proving the equivalence between (2) and (3) of Theorem \ref{possep}.  Then, \Cref{proof:IIA_implies_JDD} shows that \iia implies joint-directional differentiability; together with \Cref{proof:posterior_separable}, this shows that (1) implies (2). Finally, \Cref{proof:UHP_implies_iia} shows that the unique hyperplane property implies \iia and \ie, showing that (3) implies (1) and finishing the proof.

\subsection{Posterior-separable representation}\label{proof:posterior_separable}
Throughout this subsection, assume that the rationally inattentive preference $\succsim$ satisfies \iia and \ie. We will show that this implies that it has a posterior separable representation.

Let $\mathcal{H}$ be the set of acts that are irrelevant to the singleton
$\left\{ 0\right\} $, that is,
\begin{equation}
\mathcal{H}=\left\{ h\in\mathcal{F}:\:0\sim\left\{ 0,h\right\} \right\} .\label{eq: definition of H}
\end{equation}

\begin{lemma}\label{lem:properties of H}
For any menu $F$, we have $F\subseteq\mathcal{H}$ if and only if
$0\sim F\cup\left\{ 0\right\} $.
\end{lemma}

\begin{proof}
If $f\in F$, we always have $F\cup\{0\}\succsim \{0,f\}\succsim 0$, so $0\sim F\cup \{0\}$ implies $f\in \mathcal{H}$. The other direction can be proven by induction on the size of $F$. Let $F=\left\{ f_{1},f_{2}\ldots,f_{n}\right\} \subseteq\mathcal{H}$. If $F$
has one element, the result follows from the definition of $\mathcal{H}$. Now
let $F_{n}=F_{n-1}\cup\left\{ f_{n}\right\} $ be of size $n$ and assume that the result holds for menus of size $n-1$. Then we have $F_{n-1}\cup\left\{ 0\right\} \sim0$.
Since $f_{n}\in F_{n}\subseteq\mathcal{H}$, it follows that $0\sim\left\{ 0,f_{n}\right\} $.
By \iia, we must have $F_{n}\cup\left\{ 0\right\} =F_{n-1}\cup\left\{ 0,f_{n}\right\} \sim0$,
as we wanted.

\end{proof}
To simplify notation, we now write 
\[
\left\langle \phi,\pi\right\rangle= \int \phi \: d\pi
\]
for any integrable function $\phi$.

\begin{lemma}
The cost function $c$ is given by
\[
c\left(\pi\right)=\sup_{\begin{array}{c}
F\subseteq\mathcal{H}\\
0\in F
\end{array}}\left\langle \phi_{F},\pi\right\rangle .
\]
\end{lemma}
\begin{proof}
Let $F$ be any menu. By \ie, there exists an act $h$ such that $h\sim F\sim F\cup h$. Note that, since $\phi_h$ is an affine function, $\phi_{F\cup h}-\phi_h$ is a piecewise linear convex function. Since $u$ is surjective, there is a menu $G$, with $0\in G$, such that $\phi_G=\phi_{F\cup h}-\phi_h$. This menu $G$ satisfies two important properties: First,
\[
V(\phi_G)=V(\phi_{F\cup h}-\phi_h)=V(\phi_{F\cup h})-\phi_h(p_0)=V(\phi_h)-\phi_h(p_0)=0.
\]
This means that $0\sim G\cup 0= G$, so that $G\subseteq \mathcal{H}$.
Second, for any $\pi\in \Pi(p_0)$, we have $\left\langle \phi_{G},\pi\right\rangle =\left\langle \phi_{F\cup h}-\phi_{h},\pi\right\rangle =\left\langle \phi_{F\cup h},\pi\right\rangle -\phi_{h}\left(p_0\right)$ and therefore
\[
\left\langle \phi_{G},\pi\right\rangle -V\left(\phi_G\right)=
\left\langle \phi_{F\cup h},\pi\right\rangle -V\left(\phi_{F\cup h}\right)\geq
\left\langle \phi_{F},\pi\right\rangle -V\left(\phi_F\right),
\]
since $V(\phi_F)=V(\phi_{F\cup h})$ and $\phi_{F\cup h}\geq \phi_F$.
Therefore, by \Cref{prop:de2017rationally},
\[
c\left(\pi\right)=\sup_{
F\in \mathbb{F}}\left\langle \phi_{F},\pi\right\rangle -V\left(\phi_F\right)=\sup_{\begin{array}{c}
G\in \mathbb{F}\\
0\in G\\
V(\phi_G)=0
\end{array}}\left\langle \phi_{G},\pi\right\rangle -V\left(\phi_G\right)
=\sup_{\begin{array}{c}
G\subseteq\mathcal{H}\\
0\in G
\end{array}}\left\langle \phi_{G},\pi\right\rangle ,
\]
since $F\subseteq\mathcal{H}$ and $0\in F$ implies $V(\phi_F)=0$.
\end{proof}

Finally, we can prove the posterior separability of the cost function.
\begin{lemma}
We can write $c\left(\pi\right)=\left\langle \psi,\pi\right\rangle $,
where $\psi:\Delta\left(\Omega\right)\rightarrow\mathbb{R}$ is 
lower-semicontinuous, and given by
\begin{equation}
\psi\left(p\right)=\sup_{h\in\mathcal{H}} \sum_\omega u(h(\omega))p(\omega)\label{eq:formula for psi}
\end{equation}
 where $\mathcal{H}$ is defined as in (\ref{eq: definition of H}).
\end{lemma}
\begin{proof}
For any menu $F\subseteq\mathcal{H}$, we have $\phi_{F}\leqslant\psi$,
so that 
\[
c\left(\pi\right)=\sup_{\begin{array}{c}
F\subseteq\mathcal{H}\\
0\in F
\end{array}}\left\langle \phi_{F},\pi\right\rangle \leqslant\left\langle \psi,\pi\right\rangle .
\]
To show the converse inequality, fix $\epsilon>0$ and $\pi\in\Pi\left(\overline{p}\right)$
with support $p_{1},\ldots,p_{n}$. Suppose first that $\psi(p_i)<\infty$ for $i=1,\ldots, n$. From the definition of $\psi$,
we can find $h_{1},\ldots,h_{n}$ such that 
\[
\psi\left(p_{i}\right)<\left\langle h_{i},p_{i}\right\rangle +\epsilon \text{ for } i=1,\ldots,n.
\]
Letting $F=\left\{ 0,h_{1},\ldots,h_{n}\right\} $ we have 
\[
c\left(\pi\right)\geqslant\left\langle \phi_{F},\pi\right\rangle \geqslant\sum_{i}\left\langle h_{i},p_{i}\right\rangle \pi (p_i)>\left\langle \psi,\pi\right\rangle -\epsilon.
\]
Since $\epsilon$ was chosen arbitrarily, this shows that $c\left(\pi\right)\geqslant\left\langle \psi,\pi\right\rangle $.

Suppose now that $\psi(p_i)=\infty$ for some $i$. Then there must be a sequence $(h_n)_{n=1}^\infty$ of acts in $\mathcal{H}$ such that $\phi_{h_n}(p_i)\rightarrow \infty$. Letting $G_n=\{0,h_n\}$ we have that
\[
c(\pi)\geqslant \sup_{n\in \mathbb{N}}\left\langle \phi_{G_n},\pi\right\rangle\geqslant \sup_{n\in \mathbb{N}}\phi_{h_n}(p_i)\pi(p_i)=\infty.
\]
Lastly, recall that the epigraph of a function $f:\Delta(\Omega)\rightarrow \mathbb{R}$, denoted $\epi(f)$, is 
defined as the set of points $\{(p,z):p\in\Delta(\Omega), z\in \mathbb{R}, z\geq f(p)\}$. The epigraph of $\psi$ is given, from \eqref{eq:formula for psi}, by
\[
\epi(\psi)=\bigcap \epi(u(h))
\]
Since each function $u(h)$ is affine and hence weakly convex, and 
closure is preserved under intersection, the set $\epi(\psi)$ is 
closed as well. By \cite{Rockafellar+1970}, Theorem 7.1, $\psi$
is lower-semicontinuous.
\end{proof}

\subsection{Unique Hyperplane Property implies Joint-Directional Differentiability}\label{proof:UHP_implies_JDD}

     We prove the contrapositive: if $\psi$ does not satisfy JDD, then it does not satisfy UHP. Thus, suppose there are%
     \footnote{Recall that for any $p\in \dom(\psi)$, we define $\bar{p}=T(p)\in Y$ (see \Cref{sec:dimension}).}
     $\{p_i\}_{i=1}^K\subset \dom \psi$, $\delta\in \mathbb{R}^M\setminus \{0\}$, and $\lambda_i\in \mathbb{R}^M$ such that
\begin{enumerate}
    \item $\lambda_i\in \partial \bar{\psi} (\bar{p}_i)$,
    \item $\lambda_i+\delta \in \partial \bar{\psi} (\bar{p}_i)$,
    \item $\delta\cdot (\bar{p}_i-\bar{p}_0)=0$ for all $i$, and
    \item $p_0\in co(p_1,\ldots, p_K)$.
\end{enumerate}
In order to demonstrate the violation of UHP, we will construct a menu
for which it is optimal for the DM to choose information whose support 
is precisely $\{p_i\}_{i=1}^K$, and that the UHP will fail for this menu.

\begin{lemma}\label{lem:optimal_menu}
    Given conditions (1)-(4) above, there exists a menu $H$ such that 
    \begin{enumerate}
        \item there is an optimal distribution over posteriors for $H$ with support on $\{p_i\}_{i=1}^K$;
        \item $0,-\delta\in \Lambda_H$;
        \item $\bar{N}_H\leq 0$ and $V(\phi_H)=0$.
    \end{enumerate}
\end{lemma}
\begin{proof}
    By \Cref{lem:surjectivity}, there exists, for each $i$, an act $h_i:\Omega\rightarrow X$ such that
    \[
    \sum_{\omega\in \Omega}u(h_i(\omega))p(\omega)=\lambda_i \cdot (\bar{p}-\bar{p}_i)+\bar{\psi}(\bar{p}_i)
    \]
    for every $p\in \Delta(\Omega)$. Defining the menu $H=\{h_1,\ldots,h_K\}$, we have
\[
\phi_H(p)=\max_{h_i\in H} \sum_{\omega\in \Omega}u(h_i(\omega))p(\omega)=\max_i \lambda_i \cdot (\bar{p}-\bar{p}_i)+\bar{\psi}(\bar{p}_i).
\]
Spelling out the definition of subdifferential in condition (1), we have $\lambda_i \cdot (\bar{p}-\bar{p}_i)\leq \bar{\psi}(\bar{p})-\bar{\psi}(\bar{p}_i)$ for $i=1,\ldots, K$. This implies that $\phi_H(p)\leq \bar{\psi}(\bar{p})$ for all $p\in \Delta(\Omega)$, which in turn implies $\bar{N}_H\leq 0$. Moreover, $\bar{N}_H(\bar{p}_i)=0$ for $i=1,\ldots,K$. Since $p_0\in co(p_1,\ldots, p_K)$, this implies that $V(\phi_H)=\cav(\bar{N}_H)(\bar{p}_0)=0$, so $0\in \Lambda_H$. Moreover, if we write $p_0=\sum_{i=1}^K \beta_i p_i$ and consider the distribution over posteriors $\pi$ that puts weight $\beta_i$ on $p_i$, we will have that $\pi$ is optimal.

Similarly, using condition (3) and spelling out the definition of subdifferential in condition (2), we have, for $i=1,\ldots, K$,
\begin{align*}
    \lambda_i  \cdot (\bar{p}-\bar{p}_i)+\delta  \cdot (\bar{p}-\bar{p}_0)&=  \lambda_i  \cdot (\bar{p}-\bar{p}_i)+\delta  \cdot (\bar{p}-\bar{p}_0)+\delta \cdot (\bar{p}_0-\bar{p}_i)\\
    &=(\lambda_i +\delta) \cdot (\bar{p}-\bar{p}_i)\\
    &\leq \bar{\psi}(\bar{p})-\bar{\psi}(\bar{p}_i)
\end{align*}
This implies that $\bar{N}_H(\bar{p})\leq -\delta \cdot (\bar{p}-\bar{p}_0)$ for all $\bar{p}\in Y$. Moreover, $\bar{N}_H(\bar{p}_i)= 0=-\delta \cdot (\bar{p}_i-\bar{p}_0)$ for $i=1\ldots, K$. As before, this implies that $-\delta\in \Lambda_H$. Since $\delta \neq 0$, this shows that the unique hyperplane property is not satisfied. 
\end{proof}


\subsection{Independence of Irrelevant Alternatives implies Joint-Directional Differentiability} \label{proof:IIA_implies_JDD}

We now show that, if the preference has a posterior-separable representation and satisfies \iia, the canonical measure of uncertainty $\psi$ will satisfy JDD.

We start with a preliminary lemma regarding the subdifferentials
of sets of points that are NDISD.

\begin{lemma}
\label{compact}
    The set
    \begin{equation}
    \label{deltas}
        D(\{\bar{p}_i\}_{i=1}^K)\coloneqq\bigcap_i \left(\partial \bar{\psi}(\bar{p}_i)-\partial \bar{\psi}(\bar{p}_i)\right)\cap (\bar{p}_i-\bar{p}_0) ^\perp.
    \end{equation}
    is compact and convex. Moreover, $\bar{\psi}$ is non-differentiable in the direction $\delta$ at at $\{\bar{p}_i\}_{i=1}^K$ if and only if $\delta\in D(\{\bar{p}_i\}_{i=1}^K)$.
\end{lemma}
\begin{proof}
See \Cref{app:proof_of_compact}.
\end{proof}
\begin{lemma}
    Suppose $\succsim$ has a posterior separable representation with a $\psi$ that violates JDD. Then $\succsim$ violates \iia.
\end{lemma}

We use these properties of the set to construct a violation of IIA, 
generalizing the intuition in the example in Section 3.3. That is, when
JDD fails, one can find two different acts that, when added to the 
original menu, preserve the original respective hyperplanes, and hence
the value. However, when adding \emph{both}, the new concavification 
lies strictly higher. This is in contradiction to IIA, which states that
it must continue to provide the original value of the menu.

\begin{proof}
By Lemma \ref{compact}, the set $D(\{\bar{p}_i\}_{i=1}^K)$ is 
compact, convex, and contains a non-null linear functional. Let $\delta$ be an extreme point in this set. Note that if $\delta\in D(\{\bar{p}_i\}_{i=1}^K)$,
then so is $-\delta$.

Let $H$ be a set of acts as given by \Cref{lem:optimal_menu}, so that $0,-\delta\in \Lambda_H$.
Let $\epsilon\in \mathbb{R}^M$ be a sufficiently small vector such that
\begin{enumerate}
    \item $\bar{p}_0+\epsilon,\bar{p}_0-\epsilon\in Y$;
    \item $\delta \cdot \epsilon>0$.
\end{enumerate} 
Such an $\epsilon$ exists because the orthogonal complement of 
$D(\{p_i\}_{i=1}^K)$ does not have full rank.
Let $\alpha\in \partial \bar{\psi} (\bar{p}_0+\epsilon)$ and $\beta+\delta\in \partial \bar{\psi} (\bar{p}_0-\epsilon)$. Then
\begin{equation*}
    \alpha\cdot (\overline{p} - \overline{p_0} - \epsilon) +  \overline{\psi}(\overline{p_0} +\epsilon) - \overline{\psi}(\overline{p}) \leq 0
\end{equation*}
\begin{equation*}
    (\beta + \delta)\cdot (\overline{p} - \overline{p_0} + \epsilon) +  \overline{\psi}(\overline{p_0} - \epsilon) - \overline{\psi}(\overline{p}) \leq 0
\end{equation*}
We now define two other menus:
\begin{align*}
    F&=\{f_\alpha\}\cup H \\
    G&=\{f_\beta\}\cup H.
\end{align*}
where $f_\alpha$ is defined so that
\[
\sum_{\omega\in\Omega}u(f_\alpha(\omega))p(\omega)=\alpha\cdot (\bar{p}-\bar{p}_0-\epsilon)+\bar{\psi}(\bar{p}_0+\epsilon)
\] holds for every $p\in \dom(\psi)$, and similarly for $f_\beta$,
\[
\sum_{\omega\in\Omega}u(f_\beta(\omega))p(\omega)=\beta\cdot (\bar{p}-\bar{p}_0+\epsilon)+\bar{\psi}(\bar{p}_0-\epsilon)+\delta\cdot \epsilon.
\]
Notice that $0\in \Lambda_H$ and since  $\alpha\cdot (\bar{p}-\bar{p}_0-\epsilon)+\bar{\psi}(\bar{p}_0+\epsilon)-\bar{\psi}(\bar{p})\leqslant 0$, it remains in $\Lambda_F$, and so the concavification still yields $V(\phi_F)=0$ by (\ref{eqn:min_concavification}). Likewise, $-\delta\in \Lambda_H$ and since $\beta\cdot (\bar{p}-\bar{p}_0+\epsilon)+\bar{\psi}(\bar{p}_0-\epsilon)-\bar{\psi}(\bar{p})+\delta\cdot\epsilon\leqslant -\delta \cdot (\bar{p}-\bar{p}_0)$, it remains in $\Lambda_G$, and so the concavification still yields $V(\phi_G)=0$. Thus, all three menus\,---\,$F$,$G$, and $H$\,---\,give the same value of zero. 

Now consider $\gamma\in \mathbb{R}^M$ and $k\in \mathbb{R}$ such that $\gamma\cdot \bar{p}+k\geqslant N_{F\cup G}^* (\bar{p})$ for all $\bar{p}$. We must have 
\begin{align*}
  \gamma\cdot (\bar{p}_0+\epsilon)+k&\geqslant \alpha \cdot (\bar{p}_0+\epsilon-\bar{p}_0-\epsilon)+\bar{\psi}(\bar{p}_0+\epsilon)-\bar{\psi}(\bar{p}_0+\epsilon)=0 \\
  \gamma\cdot (\bar{p}_0-\epsilon)+k&\geqslant \beta \cdot (\bar{p}_0-\epsilon-\bar{p}_0+\epsilon)+\bar{\psi}(\bar{p}_0-\epsilon)+\delta\cdot \epsilon-\bar{\psi}(\bar{p}_0-\epsilon)=\delta\cdot \epsilon>0. 
\end{align*}
Together, these inequalities imply that $\gamma \cdot \bar{p}_0+k\geq \frac{1}{2}\delta\cdot\epsilon>0$, so whatever is the optimal hyperplane for $F\cup G$, it must get a value strictly greater than zero at $\bar{p}_0$, meaning that $V(\phi_{F\cup G})>0$. This shows that our preference violates \iia.
\end{proof}

\subsection{Unique Hyperplane Property implies \iia and \ie}\label{proof:UHP_implies_iia}

We now show that if a preference has a posterior-separable representation and satisfies the Unique Hyperplane Property, then it must satisfy \iia and \ie. Let $F$ and $G$ be menus such that $F\cap G\sim F\sim G$. Assuming that the unique hyperplane property is satisfied, we may write $\Lambda_F=\{\lambda_F\}$, $\Lambda_G=\{\lambda_G\}$, and $\Lambda_{F\cap G}=\{\lambda_{F\cap G}\}$. Since 
$F\cap G\sim F\sim G$, and the optimal posteriors for $F\cap G$ are also 
feasible for $F$ and $G$, the optimal hyperplane for $F\cap G$ must also 
be an optimal hyperplane for both $F$ and $G$. By the unique hyperplane
property, the optimal hyperplanes must therefore identical: $\lambda_F=\lambda_{F\cap G}=\lambda_G$. This implies that for 
all $f\in F\cup G$ and $p\in \Delta(\Omega)$, 
$\lambda_F\cdot (\bar{p}-\bar{p}_0)+V(\phi_F)\geqslant \sum_{\omega} u(f(\omega))p(\omega)-\psi(p)$ so that $\lambda_F\cdot (\bar{p}-\bar{p}_0)+V(\phi_F)\geqslant N_{F\cup G}^*(\bar{p})$. and so 
$\lambda_F$ is the (unique) optimal hyperplane for $F\cup G$.

To show that \ie is satisfied, let $\lambda_F$ be the unique optimal hyperplane 
for menu $F$. We will construct an act $h$ such that $\Lambda_{\{h\}}=\Lambda_{F}=\{\lambda_F\}$ and $V(\phi_F)=V(\phi_h)$ (see \cref{pic:ie}). To see why this is enough, note that we can rewrite \cref{eqn:min_concavification} as

\begin{equation*}
    V(\phi_F)=\min_{\substack{\vphantom{.}\\ \mathllap{\lambda}\in \mathrlap{\mathbb{R}^M }\\ \mathllap{v}\in \mathrlap{\mathbb{R}}}}\{v\:|\,\lambda\cdot (y-\bar{p}_0)+v\geq \bar{N}_F(y),\,\forall y\in Y \}.
\end{equation*}
Thus, if $h$ is constructed as described, we have $\lambda_F \cdot (y-\bar{p}_0))+V(\phi_F)\geq \bar{N}_h(y)$ for all $y\in Y$ and therefore we have $\lambda_F \cdot (y-\bar{p}_0))+V(\phi_F)\geq \max\{\bar{N}_h(y),\bar{N}_F(y)\}=\bar{N}_{F\cup h}(y)$ for all $y\in Y$. This means that $V(\phi_{F\cup h})\leq V(\phi_{F})$ and since the preference is monotone in inclusion, this must actually be an equality, so we have $h\sim F\sim F\cup h $, proving \ie.

To construct such an act $h$, note that we want
\[
\lambda_F \cdot (T(p)-\bar{p}_0))+V(\phi_F)\geq N_h(p)=\sum_{\omega\in \Omega}u(h(\omega))p(\omega)-\psi(p) \quad\forall p\in\dom(\psi).
\]
The left-hand side defines an affine function $A(p)\coloneqq \lambda_F \cdot (T(p)-\bar{p}_0))+V(\phi_F)$ defined on $\dom (\psi)$. By \Cref{lem:surjectivity}, there exists an act $h$ such that $\sum_{\omega\in \Omega}u(h(\omega))p(\omega)=A(p)$ for every $p\in \dom \psi$. Hence, for $p\in \dom \psi$
\begin{align*}
    N_h(p) &=\sum_{\omega\in \Omega}u(h(\omega))p(\omega)-\psi(p)\\
    &\leq \sum_{\omega\in \Omega}u(h(\omega))p(\omega)\\
    &=A(p)\\
    &=\lambda_F \cdot (T(p)-\bar{p}_0))+V(\phi_F)
\end{align*}
which implies that $\lambda_F \in \Lambda_h$ and that $V(\phi_h)=\cav (\bar{N}_h)(\bar{p}_0)\leq V(\phi_F)$. Since $\bar{N}_h(\bar{p}_0)=N_h(p_0)=V(\phi_F)$, we must actually have that $V(\phi_h)=V(\phi_F)$, showing that $h$ is the act we wanted. $\square$
\begin{figure}[t!]
\centering

\begin{tikzpicture}[scale=2, >=stealth, font=\small]

    \def\pzero{3.0}        
    \def\slope{-0.15}      
    \def\hval{2.5}         
    \def\k{0.5}            
    \def\a{1.0}            

    \pgfmathsetmacro{\shift}{(\a*\a)/(4*\k)}
    
    \pgfmathsetmacro{\offset}{\a/(2*\k)}
    \pgfmathsetmacro{\pleft}{\pzero - \offset}
    \pgfmathsetmacro{\pright}{\pzero + \offset}

    
    \tikzset{declare function={
        Hyperplane(\x) = \slope*(\x-\pzero) + \hval;
    }}

    
    \tikzset{declare function={
        Nh(\x) = Hyperplane(\x) - \k*(\x-\pzero)^2;
    }}

    \tikzset{declare function={
        Nf(\x) = Hyperplane(\x) + \a*abs(\x-\pzero) - \shift - \k*(\x-\pzero)^2;
    }}

    \draw[->, gray, thick] (-0.5,0) -- (6.5,0);
    \draw[->, gray, thick] (0,-1.0) -- (0,3.5);

    \draw[yellow!90!black, dashed, very thick] (-0.5, {Hyperplane(-0.5)}) -- (6.0, {Hyperplane(6.0)});
    \node[yellow!90!black, above] at (6.0, {Hyperplane(6.0)}) {$\lambda_F$};

    
    \draw[violet, very thick, samples=100, domain=1.0:5.0] 
        plot (\x, {Nh(\x)});
    \node[violet, right] at (5.0, {Nh(5.0)}) {$\bar{N}_h$};

    \draw[orange, very thick, samples=200, domain=0.0:6.0] 
        plot (\x, {Nf(\x)});
    \node[orange, above right] at (6.0, {Nf(6.0)}) {$\bar{N}_F$};

    
    \fill[orange] (\pleft, {Nf(\pleft)}) circle (1.5pt);
    \fill[orange] (\pright, {Nf(\pright)}) circle (1.5pt);
    \fill[violet] (\pzero, {Nh(\pzero)}) circle (1.5pt);

    \draw[dashed, cyan, thick] (\pzero, {Hyperplane(\pzero)}) -- (\pzero, 0);
    \draw[thick, cyan] (\pzero, 0.1) -- (\pzero, -0.1) node[below] {$p_0$};

\end{tikzpicture}

\caption{Constructing the ignorance equivalent given unique hyperplane}
\label{pic:ie}
\end{figure}

\subsection{Joint-directional Differentiability implies the Unique Hyperplane Property}\label{proof:JDD_implies_UHP}

Suppose the preference has a posterior-separable representation where $\psi$ is joint-directional differentiable. Let $F\coloneqq\{f_i\}_{i=1}^K$ be an arbitrary menu. For each $i$, let $a_i\in \mathbb{R}^M$ and $b_i\in \mathbb{R}$ be such that $\sum_{\omega\in\Omega}u(f_i(\omega))p(\omega)=a_i\cdot \bar{p}+b_i$ for all $p\in \dom \psi$.
Let $\{p_i\}_{i=1}^I$ be the support of an optimal distribution over posteriors for $F$, so $p_0\in co\{p_i\}_{i=1}^I$. We can assume, without loss of generality, that no act in $F$ is optimal for more than one posterior belief; for simplicity, we label the acts so that $f_i$ is optimal for $p_i$ for $i\leq K$ (relabeling any unchosen acts to have $i>K$).
Then there exists $\lambda\in \Lambda_F$ such that, for $i=1,\ldots, I$,
\[
N_F(\bar{p}_i)=a_i\cdot \bar{p}_i+b_i-\bar{\psi}(\bar{p}_i)=\lambda\cdot (\bar{p}_i-\bar{p}_0)+V(\phi_F)
\]
and
\[
a_i\cdot \bar{p}+b_i-\bar{\psi}(\bar{p})\leqslant\lambda\cdot (\bar{p}-\bar{p}_0)+V(\phi_F)
\]
for all $\bar{p}\in Y$. Subtracting the equality from the inequality, we get
\[
(a_i-\lambda)\cdot (\bar{p}-\bar{p}_i)\leqslant \bar{\psi}(\bar{p})-\bar{\psi}(\bar{p}_i),
\]
which means that $\lambda_i \coloneqq a_i-\lambda\in \partial\bar{\psi}(\bar{p}_i)$.
If there existed a second hyperplane $\hat{\lambda}$, we could repeat the same argument, and letting $\delta =\lambda-\hat{\lambda}$,
we would get that 
\[
\lambda_i+\delta=\hat{\lambda}_i\in \partial\bar{\psi}(\bar{p}_i)
\]
\[
\delta\cdot (\bar{p}_i-\bar{p}_0)=\hat{\lambda}_i\cdot (\bar{p}_i-\bar{p}_0)-\lambda_i\cdot (\bar{p}_i-\bar{p_0})
\]
\[
=[\bar{N}_F(\bar{p}_i)-V(\phi_F)]-[\bar{N}_F(\bar{p}_i)-V(\phi_F)]=0
\]
So, $\bar{\psi}$ would be NDISD at $\{\bar{p}_i\}_{i=1}^I$, contradicting 
our assumption. $\square$

\subsection{Uniqueness}

We first prove the uniqueness of $\psi$ in (2). Let $\psi$ be canonical and joint-directionally differentiable and suppose that $\tilde{\psi}$ is also canonical and represents the same preference $\succsim$. By \Cref{prop:de2017rationally}, the canonical cost function $c$ representing $\succsim$ is unique, so we must have $c(\pi)=\int \psi d\pi=\int \tilde{\psi} d\pi$ for every $\pi \in \Pi(p_0)$. Joint-directional differentiability of $\psi$ implies its differentiability in the directions of its effective domain at $p_0$ (since $p_0\in co\{p_0\}$). By \Cref{prop:canonps}, $\psi$ is the unique canonical measure of uncertainty that is a representation for $\psi$, so $\tilde{\psi}=\psi$.

To prove uniqueness in (3), note that, by \Cref{proof:UHP_implies_JDD}, any $\psi$ that satisfies UHP must satisfy JDD, so the uniqueness follows from the same argument above.

\section{Discussion of Main Theorem}

As was noted when we introduced canonical costs, it is without loss 
to write $c(\pi)=\sup_{\psi\in\Psi}\int \psi d\pi$. For this 
discussion, we explore the implications of our axioms when the set 
$\Psi$ is finite.\footnote{Notice that, by definition of supremum, one can 
approximate the cost function $c$ by a finite set 
$\hat{\Psi}\subset \Psi$, defining a cost function 
$\hat{c}(\pi)\equiv\max_{\psi\in\hat{\Psi}}\int \psi d\pi$,
as follows. In a case where $c(\pi)$ is finite, then there exists a 
sufficiently large, finite $\hat{\Psi}\subset \Psi$ such that 
$c(\pi)<\hat{c}(\pi)+\epsilon$. if $c(\pi)=\infty$, then for sufficiently
large $\hat{\Psi}$, $\pi$ is dominated by $\delta_{p_0}$.} Then, we not only have $\psi(p_0)\leq 0$ for all $\psi\in\Psi$, but also have equality for at least one $\psi\in \Psi$. We may also assume that $\Psi$ is minimal.

For any finite menu $H$, let $\pi_H\in \Pi(p_0)$ be an optimal information 
choice given that menu. Define 
$\psi_H\in\arg\max_{\psi\in\Psi}\int\psi d\pi_H$; by the minimality of 
$\Psi$, there will be at least one menu $H$ for which $\psi_H$ is the 
unique maximizer.

\ie, Ignorance Equivalence, comes to rule out that there exist $\psi^*\in\Psi$ such that $\psi^*(p_0)<0$. To see this, suppose that 
such a $\psi^*\in\Psi$ exists. Let $\Psi_0=\{\psi\in \Psi|\psi(p_0)=0\}$, which is nonempty by \Cref{prop:dmr2021}. By the minimality of $\Psi$, there is some menu $F^*$ such that 
$c(\pi_{F^*})=\int \psi^* \:d\pi_{F^*}$ and, for all $\psi\neq \psi^*$, 
$c(\pi_{F^*})>\int \psi \:d\pi_{F^*}$. Let 
$\delta\coloneqq \min\{-\psi^*(p_0),\min_{\psi\neq \psi^*} c(\pi_{F^*})-\int \psi \:d\pi_{F^*}\}>0$. 

For any $h$ such that $h\sim F^*$, one will have 
$\pi_h=\delta_{p_0}$ and hence $\psi_h\in\Psi_0$. Consider now the menu 
$F^*\cup \{h\}$. By choosing 
$\pi^\prime \coloneqq\frac{1}{2}\pi_{F^*}+\frac{1}{2}\delta_{p_0}$, 
one can save on information costs: from the fact that 
$\psi^*\notin \Psi_0$ and the representation of $c$ is convex as in 
\Cref{prop:dmr2021}, there exists $\psi^\prime\in \Psi$ such that 
\[
c(\pi^\prime)=\int \psi^\prime d\pi^\prime
\]
\[
\leq \frac{1}{2}\int \psi^*d\pi_{F^*}+\frac{1}{2}\psi_h(p_0)-\frac{1}{2}\delta
\]
\begin{equation}
\label{iecostsave}
\leq \frac{1}{2}c(\pi_{F^*})-\frac{1}{2}\delta
\end{equation}
However, by the monotonicity of the indirect utility from the decisions in
the menu size, due to the preference for flexibility,
\begin{equation}
\label{ieusame}
    \int \phi_{F^*\cup\{h\}} d\pi^\prime\geq \frac{1}{2}\int \phi_{F^*}d\pi_{F^*}+\frac{1}{2}\phi_h(p_0),
\end{equation}
Therefore, 
\[
V(\phi_{F^*\cup\{h\}})\geq \int \phi_{{F^*\cup\{h\}}}d\pi^\prime -c(\pi^\prime)
\]
\[
\geq \frac{1}{2}\int \phi_{F^*}d\pi_{F^*}+\frac{1}{2}\phi_h(p_0)-c(\pi^\prime)
\]
\[
\geq \frac{1}{2}[\int \phi_{F^*}d\pi_{F^*}-\int \psi^*d\pi_{F^*}]-\frac{1}{2}[\phi_h(p_0)]+\frac{1}{2}\delta
\]
\[
=V(\phi_{F^*})+\frac{1}{2}\delta
\]
where the last equality is from $h\sim F^*$. So, $F^*\cup \{h\}\succ F^*$, violating \ie.

On the other hand, there is no problem with \ie if $\Psi=\Psi_0$. 
In this case, for a given $\psi_{F^*}$, the act $h$ is consistent with 
$\psi_{F^*}$ itself, and so there are no savings of information costs to 
be had by randomizing.

The presence of \ie allows for  \iia to have bite. Without \ie, 
it is unclear whether there will actually be any menus $F,G$ with 
$\psi_F\neq\psi_G$ such that $F\cap G\neq \emptyset$; if that were so, 
then in such cases, \iia holds vacuously for $F,G$. However, with \ie,
we know that $\psi_F(p_0)=\psi_G(p_0)=0$. By \ie, there exist $h,h^\prime$
such that $h\sim F\sim F\cup\{h\}$ and 
$h^\prime\sim G\sim G\cup\{h^\prime\}$. Since adding an affine function 
$\alpha$ to the payoffs of all acts in the menu does not change the 
optimal choice of information, and changes the value of the menu by
$\alpha\cdot p_0$, one can replace $G\cup\{h^\prime\}$ with 
$\hat{G}\cup \{h\}$ by letting $\alpha=h-h^\prime$, such that 
$\psi_G=\psi_{\hat{G}}$ and $F\sim \hat{G}$. Thus it is without loss
to consider $G=\hat{G}$. As a result, under \ie, if there are 
$\psi_F\neq \psi_G$, then it is there are corresponding $F,G$ such 
$h\sim F\sim F\cup \{h\}\sim G\cup \{h\}\sim G$. Note that 
$\psi_F=\psi_{F\cup\{h\}}$ and $\psi_G=\psi_{G\cup\{h\}}$ since 
$\pi_F$ and $\pi_G$ are optimal information choices for $F\cup\{h\}$
and $G\cup\{h\}$, respectively.

Given such $F,G$, then by taking $F\cup G\cup \{h\}$, \iia dictates that 
$F\sim F\cup G\cup \{h\}$. However, if $\psi_F\neq\psi_G$, this is not 
the case: one can save on the information cost by randomizing 
$\frac{1}{2}\pi_F+\frac{1}{2}\pi_G$, while keeping the expected utility
from the decisions the same, similarly to the cost 
saving/expected utility preservation due to \ie in equations
(\ref{iecostsave}) and (\ref{ieusame}). So, $\psi_F=\psi_G$.

\section{Uniform Posterior-Separable Costs}

We now allow for the agent's prior to vary. For each possible prior $p_0\in\Delta(\Omega)$, we denote the cost of information by $c_{p_0}:\Pi(p_0)\rightarrow \mathbb{R}\cup \{\infty\}$. For each $\pi\in \Delta(\Delta(\Omega))$, we can deduce the corresponding prior $p_0$ by taking the expectation of $\pi$, so that $\pi\in \Pi(p_0)$. Thus, we may write the cost function for all priors under a single notation $c:\Pi\rightarrow \mathbb{R}$. We call $c$ canonical if, for each $p_0$, $c_{p_0}$ is canonical. We now consider the following condition that allows for the costs to be ``the same" across menus.

\begin{definition}\label{def:prior_invariant}
    The canonical cost function $c:\Delta(\Delta(\Omega))\rightarrow \mathbb{R}\cup \{\infty\}$ is \emph{prior invariant} with respect to the set $\Psi$ of convex functions that \emph{represents} $c$ if
    \[
    c(\pi)=\sup_{\psi\in \Psi} \int \psi \:d\pi-\sup_{\psi'\in \Psi} \psi'\left(\int p \:d\pi\right).
    \]
    When $\Psi$ is a singleton, we say that $c$ is \emph{uniformly posterior separable}.
\end{definition}

Note that if $c$ is canonical then, by \Cref{prop:dmr2021}, for every $p_0$ there exists a set of convex functions $\Psi_{p_0}$ such that 
\[
c(\pi)=\sup_{\psi\in\Psi_{p_0}} \int \psi \: d\pi=\sup_{\psi\in \Psi_{p_0}} \int \psi \:d\pi-\max_{\psi'\in \Psi_{p_0}} \psi'(p_0).
\]
Prior-invariance imposes that the same set $\Psi$ applies to every $p_0$. Note that we need to subtract the second term to guarantee that $c_{p_0}$ is grounded for every $p_0$.

The quintessential uniformly posterior-separable cost is mutual information. Denoting the entropy of a distribution by $H(p)=-\sum_{\omega}p(\omega)\ln(p(\omega))$, mutual information is given by the expected reduction in entropy,
\[
c_{p_0}(\pi)=H(p_0)-\int H(p) \:\pi(dp),
\]
which matches the formula in \Cref{def:prior_invariant} by setting $\Psi=\{-H\}$.

\begin{theorem}
For a fixed prior $p_0$, let $\Psi$ be a set of convex functions representing $c_{p_0}$ satisfying the conditions of \Cref{prop:dmr2021}. Let the canonical cost function $c$ be prior-invariant with respect to $\Psi$. If, for each prior, the preference over menus corresponding to $c$ satisfies \iia and \ie, then $c$ is uniformly posterior separable and represented by some $\psi:\Delta(\Omega)\rightarrow\mathbb{R}\cup \{\infty\}$ that is differentiable on $ \ri(\dom(\psi))$.
\end{theorem}

\begin{proof}
    By \Cref{possep}, the cost function $c_{p_0}$ is posterior separable. Since $\Psi$ is minimal for $c_{p_0}$, we must have $\Psi=\{\psi\}$ for some convex $\psi:\Delta(\Omega)\rightarrow \mathbb{R}\cup \{\infty\}$. Since $c$ is prior-invariant, we must have that, for all priors $q\in \Delta(\Omega)$ and $\pi\in \Pi(q)$,
    \[
    c_{q}(\pi)=\sup_{\psi'\in \Psi}\int \psi'(p)d\pi-\sup_{\psi''\in \Psi} \psi''\left(q\right)=\int \psi(p)d\pi-\psi(q),
    \]
    so $c$ is uniformly posterior separable.

    To see that $\psi$ must be differentiable, let $q\in \ri(\dom(\psi))$ and fix the prior to be $q$.  By \Cref{possep}, the preference $\succsim$, satisfies \iia and \ie, so it has a posterior separable representation with a canonical measure of uncertainty $\hat{\psi}$ that satisfies joint-directional differentiability. Thus for every $\pi\in \Pi(q)$, 
    \[
    \int \hat{\psi}\: \pi(dp)=\int \psi\: d\pi-\psi(q)=\int [\psi(p)-\psi(q)]\: \pi(dp)= \int \tilde{\psi}(p)\:\pi(dp).
    \]
    where $\tilde{\psi}(p)=\psi(p)-\psi(q)$. By \Cref{lem:affine_rep} in \Cref{app:proof_of_canonps}, there exists $\xi$ such that
    \[
    \hat{\psi}(p)=\tilde{\psi}(p)+\xi\cdot (p-q).
    \]
    Since $\hat{\psi}$ satisfies joint-directional differentiability, it must be differentiable at $q$ (simply pick $K=1$ and $\{p_i\}_{i=1}^K=\{q\}$ and we trivially have $q\in co(\{q\})$). This implies that $\tilde{\psi}$ is differentiable at $q$ and so is $\psi$.
\end{proof}

\section{Related Literature and Discussion}

Posterior-separable costs have been widely used, primarily because the optimal hyperplane formulation allows practical solutions. For instance,
\cite{lipnowski2020attention} considers the distortion of attention when initial 
information is provided to the decision maker. Several models consider mechanism design with inattentive buyers, including 
 \cite{mensch2022screening}, \cite{yoder2022designing}, and 
\cite{gleyze2023informationally}. \cite{yang2020optimality} considers 
the optimal structure of a security for sale when the buyer is rationally 
inattentive. Lastly, \cite{hebert2023information} considers a rational expectations model with inattentive agents. 
Our behavioral characterization of posterior separable costs, which leads to a unique optimal hyperplane, provides a different perspective on the assumptions behind these applications.

There is also a literature that investigates whether a finite dataset is consistent with costly information acquisition through
revealed preference, in the spirit of \cite{afriat1967construction}. This literature includes \cite{caplin2015revealed}, \cite{chambers2020costly}, and \cite{mensch2024posterior}, but the closest to our paper is  
\cite{denti2022posterior}, which tests for posterior 
separable costs.\footnote{This literature assumes as given the state-dependent stochastic choice implied by the rational inattention model on a finite number of menus. In theory, to perfectly identify such distributions, one would need an infinite number of observations of choices over the same menu. In practice, one repeats each choice a finite number of times and takes the empirical distribution, assuming it is a good enough approximation.} Although his representation is almost the same as ours (with joint-directional differentiability being the only difference), the exercise in this paper is of a different nature. Here, we take as given the preference over all menus of acts, as in \cite{de2017rationally}. This richer domain allows us to isolate the behavioral implications of the model, and to study identification and uniqueness issues.



The search for a unique representation naturally led to our definition of a \emph{canonical} measure of uncertainty, including the new condition that the prior lie in the relative interior of its effective domain (see the discussion following \Cref{def:canonical}). However, note that imposing that the measure of uncertainty $\psi$ be canonical, by itself, is not enough to achieve uniqueness\,---\,if $\psi$ is not differentiable at $p_0$, it is possible to find an $\xi\in \mathbb{R}^\Omega\setminus\{0\}$ such that $\psi(p)+\xi\cdot (p-p_0)$ is still canonical; this change in $\psi$ does not alter the overall cost of information $c(\pi)$ (and hence does not alter the preference; see also \Cref{lem:affine_rep} in \Cref{app:proof_of_canonps}). So, to uniquely identify the measure of uncertainty, one also needs differentiability at $p_0$, which in our case is implied by joint-directional differentiability.

\bibliographystyle{jpe.bst}
\bibliography{PSmenucosts}

\newpage
\appendix
\section{Appendix}
\subsection{Proof of \Cref{prop:dmr2021}}\label{app:proof_dmr2021}

By \cite{de2017rationally}, Theorem 2, the cost of information
    can be written as 
    \[
    c(\pi)=\sup_{F\in\mathbb{F}} \int \phi_F(p)d\pi(p)-V(\phi_F)
    \]
Notice that for each $F$, $\phi_F-V(\phi_F)$ is a convex function, and so one can let $\Psi=\{\phi_F-V(\phi_F):\,F\in\mathbb{F}\}$ to establish 
\[
c(\pi)=\sup_{\psi\in \Psi} \int \psi d\pi
\]
Write $\Psi'\sim \Psi$ to denote the equivalence relation that says that $\Psi'$ and $\Psi$ represent the same cost function $c$. We now show that there exists a $\Psi'\sim \Psi$ satisfying properties 1 and 2 in the proposition. To do so, we divide in two cases:
\begin{enumerate}
    \item Suppose that there exists a finite $\Psi'\sim \Psi$. If $\Psi'$ is minimal, we are done; if not, there exists a $\Psi''$ with $|\Psi''|<|\Psi|$ such that $\Psi''\sim \Psi' \sim \Psi$. We can then repeat the same argument, until we reach a minimal $\Psi^*$, which must happen in a finite number of steps since $\Psi'$ was finite. Since $\Psi^*$ is finite, the maximum must always be achieved, and in particular $c(\delta_{p_0})=\max_{\psi\in \Psi^*}\psi(p_0)=0$.

    \item Suppose that there is no finite $\Psi'\sim \Psi$. Then $\Psi$ is already minimal. To satisfy condition (1), note that, since $c$ is grounded,
    \[
    0=c(\delta_{p_0})=\sup_{\psi\in \Psi} \psi(p_0),
    \]
    and since $c$ is also Blackwell monotone, $c(\pi)\geqslant 0$ for all $\pi\in \Pi(p_0)$. This means that if we add $\psi_0\equiv 0$ to $\Psi$, we get 
    \[
    \sup_{\psi\in \Psi\cup \{\psi_0\}} \int \psi \:d\pi=\max\left\{ \sup_{\psi\in \Psi} \int \psi \:d\pi,0\right\}=\max\{c(\pi),0\}=c(\pi)
    \]
    so $\Psi\cup \{\psi_0\}\sim \Psi$, and
    \[
    c(\delta_{p_0})=\psi_0(p_0)=0=\max_{\psi\in \Psi\cup \{\psi_0\}} \psi(p_0),
    \]
    which is condition (1). Since $|\Psi\cup \{\psi_0\}|= |\Psi|=\infty$, condition (2) is also satisfied by $\Psi\cup \{\psi_0\}$.
\end{enumerate}

\subsection{Proof of \Cref{prop:canonps}}\label{app:proof_of_canonps}
Let $c$ be canonical and represented by $\psi$. 

\noindent (\ref{psi_convex}) It follows from Blackwell monotonicity that $\psi$ must be convex (see \cite{lipnowski2022predicting}, Lemma 24). 

\noindent (\ref{psi_grounded}) From groundedness, it is immediate that $\psi(p_0)=0$. 

\noindent (\ref{ri_prior}) We now construct a $\hat{\psi}$ that represents the same cost of information as $\psi$, satisfying $p_0\in \ri (\dom \hat{\psi})$. If $p_0\in \ri (\dom \psi) $, we simply let $\hat{\psi}=\psi$. So suppose $p_0\notin \ri (\dom \psi) $ and let $M$ denote the dimension of $\aff(\dom\psi)$. Since $p_0$ is on the boundary of $\dom\psi$, by 
    \cite{Rockafellar+1970}, Theorem 11.6, there exists a supporting
    hyperplane, defined by some vector $\lambda$, such that, letting
    \[
    X\coloneqq\{p\in \Delta(\Omega):\,\lambda\cdot p\geq \lambda \cdot p_0\}
    \]
    we have that for all $p\in X$, 
    \[
    \lambda\cdot p>\lambda \cdot p_0\implies \psi(p)=\infty,
    \]
    while $X\cap \ri(\dom(\psi))=\emptyset$, i.e. $X$ only intersects with $\dom(\psi)$ on the boundary of the latter.

Now let 
\begin{equation*}
    \hat{\psi}(p)=
    \begin{cases}
      \psi(p), & \text{if}\ p\in X \\
      \infty, & \text{otherwise}
    \end{cases}
  \end{equation*}
That is, $\hat{\psi}$ replaces $\psi(p)$ with $\infty$ for all $p\notin X$. In particular, $\hat{\psi}(p)=\infty,\,\forall p\in\ri(\dom(\psi))$.
We shall show that $c(\pi)=\int \psi \:d\pi=\int \hat{\psi} \:d\pi$ for every $\pi\in \Pi(p_0)$. Indeed, if $\mbox{supp}(\pi)\subset X$, this is obviously true. When $\mbox{supp}(\pi)\not\subset X$,
by Bayes' rule then, with positive probability according to $\pi$,  there is a $p\in \mbox{supp}(\pi)$ such that $\lambda\cdot p>\lambda\cdot p_0$, and therefore $\psi(p)=\infty$. 
Consequently, $c(\pi)=\infty = \int \hat{\psi}\: d\pi$. 

If $p_0\in \ri(\dom(\hat{\psi}))$, we found our desired $\hat{\psi}$. Otherwise, note that $\dim(\dom(\hat{\psi}))<\dim(\dom(\psi))$ and repeat the procedure until $p_0\in \ri(\dom(\hat{\psi}))$ or the dimension is zero (in which case the effective domain is just $\{p_0\}$ and the cost of information was trivial).

\noindent (\ref{psi_positive}) Let $T:\aff(\dom \hat{\psi})\rightarrow \mathbb{R}^M$ be a bijective affine transformation as in \Cref{sec:dimension}. Since $\hat{\psi}$ is convex, so is $T^*\hat{\psi}$, so there exists a $\lambda\in \partial T^*\hat{\psi}$, that is, $\lambda \cdot (\bar{p}-\bar{p}_0)\leqslant T^*\hat{\psi}(\bar{p})=\hat{\psi}(p)$ for all $p\in \dom \hat{\psi}$. Let $\tilde{\psi}(p)=\hat{\psi}(p)-\lambda\cdot (\bar{p}-\bar{p}_0)$. Then $\tilde{\psi}\geqslant 0$ and it inherits properties (\ref{psi_convex}), (\ref{psi_grounded}), and (\ref{ri_prior}) from $\hat{\psi}$, so $\tilde{\psi}$ is canonical. Also note that, for all $\pi\in \Pi(p_0)$, 
\[
\int \tilde{\psi}\: d\pi= \int \left[ \hat{\psi}(p)-\lambda \cdot (T(p)-T(p_0))\right]\: \pi(dp)=\int \hat{\psi}\: d\pi-\lambda \cdot (T(p_0)-T(p_0))=c(\pi).
\]

\noindent(\ref{psi_lsc}) As $\hat{\psi}=\psi$ on $\aff(\dom\hat{\psi})$ and $\infty$ otherwise, the lower-semicontinuity of $\hat{\psi}$ follows from the lower-semicontinuity of $\psi$
restricted to $\aff(\dom\hat{\psi})$. Since lower-semicontinuity
is preserved by affine transformation, $\psi^\prime$ is lower-semicontinuous as well.

Before proving the uniqueness of the canonical representation, we prove the following lemma. 

\begin{lemma}\label{lem:affine_rep}
    If $\psi$ and $\psi'$ be convex measures of uncertainty satisfying
    \begin{enumerate}
        \item $\psi(p_0)=\psi'(p_0)=0$;
        \item $p_0\in \ri\dom \psi \cap \ri\dom \psi'$;
        \item $\int \psi \:d\pi=\int \psi' \:d\pi$ for all $\pi\in\Pi(p_0)$.
    \end{enumerate}
    Then there exists a $\xi\in \mathbb{R}^{\vert \Omega\vert}$ such that $\psi(p)=\psi'(p)+\xi\cdot (p-p_0)$ for all $p\in \Delta(\Omega)$.
\end{lemma}
\begin{proof}
    We first prove that $\dom \psi=\dom \psi'$. Let $p\in \dom \psi$ be arbitrary. Then $p_0+t(p-p_0)\in \aff (\dom \psi)$ for all $t\in \mathbb R$. Since $p_0\in \ri (\dom\psi)$, we can find an $\epsilon>0$ small enough that $p'\coloneqq p_0-\epsilon(p-p_0)\in \dom \psi$. Then $p_0$ is in the convex hull of $\{p,p'\}$, which means we can find a $\pi\in \Pi(p_0)$ with support $\{p,p'\}\subset \dom \psi$. Then $c(\pi)=\int \psi \: d\pi<\infty$, which means that $\int \psi' \: d\pi<\infty$ as well, which can only happen if $p\in \dom \psi'$. Hence $\dom \psi\subset\dom \psi'$ and, by symmetry, $\dom \psi=\dom \psi'$. This proves that $\psi(p)=\psi'(p)+\xi\cdot (p-p_0)$ for all $p\notin \dom \psi$ regardless of $\xi$ since the finite term $\xi\cdot (p-p_0)$ becomes irrelevant. 

    Now, for $p\in \dom \psi$, let $\zeta(p)=\psi(p)-\psi'(p)$. Since $\psi(p_0)=\psi'(p_0)=0$, $\zeta(p_0)=0$. By Theorem 1.5 in \cite{Rockafellar+1970}, the proof will be finished if we show that $\zeta$ is an affine function. To that end, let $\alpha\in[0,1]$ and $x,y\in \dom\psi$. Suppose first that $\alpha x+(1-\alpha)y=p_0$. Then, letting $\pi\in \Pi(p_0)$ have support $\{x,y\}$, we get
    \[
    0=\int \zeta \:d\pi= \alpha \zeta(x)+(1-\alpha)\zeta(y).
    \]
    Since $\zeta(\alpha x+(1-\alpha)y)=\zeta(p_0)=0$, we have $\zeta(\alpha x+(1-\alpha)y)=\alpha \zeta(x)+(1-\alpha)\zeta(y)$.
    
    Now suppose that $\alpha x+(1-\alpha)y\neq p_0$. Let $\epsilon>0$ be small enough that $z=p_0-\epsilon(\alpha x+(1-\alpha)y-p_0)\in \dom \psi$ (recall that $p_0\in \ri\dom \psi).$ Letting $t=(1+\epsilon)^{-1}\in (0,1)$, we have that $p_0=tz+(1-t)(\alpha x+(1-\alpha)y)$. Thus there is a $\pi\in \Pi(p_0)$ with support on $\{z,\alpha x+(1-\alpha)y\}$ so that 
    \begin{equation}\label{eq:zeta_1}
    0=\int \zeta \:d\pi=t\zeta(z)+(1-t)\zeta(\alpha x+(1-\alpha)y).        
    \end{equation}
    Similarly, we may write $p_0=tz+(1-t)\alpha x+(1-t)(1-\alpha)y$ so there is a $\pi'\in \Pi(p_0)$  with support $\{x,y,z\}$ putting probability $t$ on $z$. Hence
    \begin{equation}\label{eq:zeta_2}
    0=\int \zeta \:d\pi' =t\zeta(z)+(1-t)\alpha \zeta(x)+(1-t)(1-\alpha)\zeta(y).        
    \end{equation}
    Putting together equations \cref{eq:zeta_1} and \cref{eq:zeta_2}, we get that $\zeta(\alpha x+(1-\alpha)y)=\alpha \zeta(x)+(1-\alpha)\zeta(y)$. Thus, we have shown that $\zeta$ is affine, which finishes the proof. 
\end{proof}

Finally, if $\psi$ is differentiable in the directions of its effective domain at $p_0$, so are $\hat{\psi}$ and $\tilde{\psi}$, by construction. If $\psi'$ were another canonical measure of uncertainty representing the same cost of information, then, by \Cref{lem:affine_rep}, there would be a $\xi$ such that $\tilde{\psi}(p)=\psi'(p)+\xi\cdot (p-p_0)$. Also, by the proof of \Cref{lem:affine_rep}, $\tilde{\psi}$ and $\psi'$ must have the same domain, so we may use the same bijective affine transformation $T$ for both. Since $T^*\psi'(\bar{p})\geqslant 0$ and $T^*\tilde{\psi}(\bar{p}_0)=0$, this would imply that $T^*\tilde{\psi}(\bar{p})-T^*\tilde{\psi}(\bar{p}_0)\geqslant \bar{\xi}\cdot (\bar{p}-\bar{p_0})$, or $\bar{\xi}\in \partial T^*\tilde{\psi}$. But since $\tilde{\psi}$ is differentiable in the directions of its effective domain at $p_0$ and and $T^* \tilde{\psi}(\overline{p}) - T^* \tilde{\psi}(\overline{p_0}) = T^* \tilde{\psi}(\overline{p}) \geq 0$,  $\bar{\xi}=0$ must be the only element in the subdifferential, meaning that $T^*\tilde{\psi}=T^*\psi'$, which implies that $\tilde{\psi}=\psi'$. 
\subsection{Proof of \Cref{compact}}\label{app:proof_of_compact}

    Notice that $0\neq \delta \in D(\{\bar{p}_i\}_{i=1}^K)$ if and only if $\delta\in \partial \bar{\psi}(\bar{p}_i)-\partial \bar{\psi}(\bar{p}_i)$ and $\delta \cdot (\bar{p}_i-\bar{p}_0)=0$ for $i=1,\ldots,K$, which is equivalent to the definition of  $\psi$ being non-differentiable in the direction $\delta$. By \citet[Theorem 23.2]{Rockafellar+1970}, $\partial\bar{\psi}(\bar{p}_i)$ 
    is closed and convex; since this is preserved under subtraction 
    and intersection, so is $D(\{\bar{p}_i\}_{i=1}^K)$. To show that $D(\{\bar{p}_i\}_{i=1}^K)$ is compact, it remains to show that $\partial\bar{\psi}(\bar{p}_i)$ is bounded.
    By condition (2) of non-differentiability in the same direction,
    we have $\delta(\bar{p}_i-\bar{p}_0)=0$. Using \Cref{lem:optimal_menu}, fix a menu $H$ for which the 
    posteriors $\{p_i\}_{i=1}^K$ are optimal. Notice that in this case, 
    if $\lambda,\lambda^\prime\in \Lambda_H$, then, since for 
    $\hat{\lambda}\in \{\lambda,\lambda^\prime\}$
    \[
    \hat{\lambda}\cdot (\bar{p}_i-\bar{p}_0)+V(\phi_H)=\bar{N}_H(\bar{p}_i)\quad \forall i,
    \]
    it follows that 
    $(\lambda-\lambda^\prime)\cdot (\bar{p}_i-\bar{p}_0)=0$.
    Therefore,
    \[
    D(\{\bar{p}_i\}_{i=1}^K)=\Lambda_H-\Lambda_H
    \]
    and so $D(\{\bar{p}_i\}_{i=1}^K)$ is compact if $\Lambda_H$ is.
    Then, since
    $\delta\cdot (\bar{p}_i-\bar{p}_0)=0$, $V(\phi_H)=\lambda\cdot (\bar{p}_0-\bar{p}_i)+\bar{N}_F(\bar{p}_i)= [\lambda+\delta]\cdot (\bar{p}_0-\bar{p}_i)+\bar{N}_F(\bar{p}_i)$,
    for all $i\in\{1,...,K\}$, and therefore $\lambda+\delta$ is also an optimal hyperplane. By way of contradiction, suppose 
    $\Lambda_H$ were not compact, then by \citet[Theorem 8.4]{Rockafellar+1970}, 
    there would exist $\lambda,\delta\neq 0$ such that, for all $t>0$, 
    $\lambda+t\delta\in \Lambda_H$. Furthermore, as 
    $D(\{\bar{p}_i\}_{i=1}^K)$ has dimension at least $1$, its orthogonal 
    complement does not have full rank \citep[p. 5]{Rockafellar+1970}. As a 
    result, for all $\epsilon>0$, there would exist $p$ such that 
    $\Vert \bar{p}-\bar{p}_0\Vert<\epsilon$, $\delta\cdot (\bar{p}-\bar{p}_0)< 0$, and 
    $\phi(p)-\psi(p)\leq V(\phi_H)+[\lambda+t\delta]\cdot (\bar{p}-\bar{p}_0),\,\forall t>0$, which can only happen if $\psi(p)=\infty$. Thus, there would be a sequence $\{\hat{p}_j\}_{j=1}^\infty\subset \dom(\psi)$ such that $\lim_{j\rightarrow \infty} \hat{p}_j=p_0$ and 
    $\psi(\hat{p}_j)=\infty,\forall j$, contradicting the assumption that 
    $p_0\in \ri(\dom(\psi))$ for $\psi$ canonical with 
    $\dim(\dom(\psi))=M$.

\subsection{Representations of affine functions via acts}

\begin{lemma}\label{lem:surjectivity}
    Let $D\subseteq \Delta(\Omega)$ be a convex set and $A:D\rightarrow \mathbb{R}$ be an affine function. There exists an act $h:\Omega\rightarrow X$ such that, for every $p\in D$, 
    \[
    \sum_{\omega}u(h(\omega))p(\omega)=A(p).
    \]
\end{lemma}
\begin{proof}
    $A$ can be extended to an affine function with domain $\aff (\Delta(\Omega))=\{x\in \mathbb{R}^\Omega|\sum_\omega x(\omega)=1\}$ (see Theorem 1.4 and page 7 of \cite{Rockafellar+1970}), which can then be extended to a linear function $L:\mathbb{R}^\Omega \rightarrow \mathbb{R}$, which can be uniquely represented as $L(p)=\sum_{\omega\in \Omega} l(\omega)p(\omega)$. Since the image of $u$ is $\mathbb{R}$, we can find an act $h$ such that $u(h(\omega))=l(\omega)$ for all $\omega\in \Omega$. Thus, we have, for every $p\in D$,
    \[
    \sum_{\omega}u(h(\omega))p(\omega)=\sum_{\omega}l(\omega)p(\omega)=L(p)=A(p).
    \]
\end{proof}

\end{document}